\documentclass[a4paper,UKenglish,cleveref, autoref]{lipics}
\usepackage{enumitem}
\setlist[enumerate]{label*=\arabic*.}

\theoremstyle{definition}
\newtheorem{conjecture}[theorem]{Conjecture}

\newcommand{\Patrascu}{P\v{a}tra\c{s}cu}
\newcommand{\modulo}{\operatorname{mod}{}}

\author[1]{Isaac Goldstein\thanks{This research is supported by the Adams Foundation of the Israel Academy of Sciences and Humanities}}
\author[1]{Moshe Lewenstein\thanks{This work was partially supported by ISF grant \#1278/16}}
\author[1]{Ely Porat \thanks{This work was partially supported by ISF grant \#1278/16 and ERC grant MPM - 683064.}}
\affil[1]{Bar-Ilan University, Ramat Gan, Israel \\  \texttt{\{goldshi,moshe,porately\}@cs.biu.ac.il}}

\authorrunning{I. Goldstein, M. Lewenstein and E. Porat}
\Copyright{Isaac Goldstein, Moshe Lewenstein and Ely Porat}

\begin{document}

\title{On the Hardness of Set Disjointness and Set Intersection with Bounded Universe}
\date{}

\maketitle

\thispagestyle{empty}

\begin{abstract}
In the SetDisjointness problem, a collection of $m$ sets $S_1,S_2,...,S_m$ from some universe $U$ is preprocessed in order to answer queries on the emptiness of the intersection of some two query sets from the collection. In the SetIntersection variant, all the elements in the intersection of the query sets are required to be reported. These are two fundamental problems that were considered in several papers from both the upper bound and lower bound perspective.

Several conditional lower bounds for these problems were proven for the tradeoff between preprocessing and query time or the tradeoff between space and query time. Moreover, there are several unconditional hardness results for these problems in some specific computational models. The fundamental nature of the SetDisjointness and SetIntersection problems makes them useful for proving the conditional hardness of other problems from various areas. However, the universe of the elements in the sets may be very large, which may cause the reduction to some other problems to be inefficient and therefore it is not useful for proving their conditional hardness.

In this paper, we prove the conditional hardness of SetDisjointness and SetIntersection with bounded universe. This conditional hardness is shown for both the interplay between preprocessing and query time and the interplay between space and query time. Moreover, we present several applications of these new conditional lower bounds. These applications demonstrates the strength of our new conditional lower bounds as they exploit the limited universe size. We believe that this new framework of conditional lower bounds with bounded universe can be useful for further significant applications.
\end{abstract}

\newpage

\section{Introduction}
The emerging field of fine-grained complexity receives much attention in the last years. One of the most notable pillars of this field is the celebrated 3SUM conjecture. In the 3SUM problem, given a set of $n$ numbers we are required to decide if there are 3 numbers in this set that sum up to zero. It is conjectured that no truly subquadratic solution to this problem exists. This conjecture was extensively used to prove the conditional hardness of other problems in a variety of research areas, see e.g.~\cite{ACLL14,AKLPPS16,AL13,AW14,AWW14,AWY15,BHP99,GKLP16,GO95,KPP16,Patrascu10,WW13}. The 3SUM problem is closely related to the fundamental SetDisjointness problem. In the SetDisjointness problem we are given $m$ sets $S_1,S_2,...,S_m$ from some universe $U$ for preprocessing. After the preprocessing phase, given a query pair of indices $(i,j)$ we are required to decide if the intersection $S_i \cap S_j$ is empty or not. In the SetIntersection variant, all the elements within the intersection $S_i \cap S_j$ are required to be reported.

Cohen and Porat~\cite{CP10} investigated the upper bound of both problems. Specifically, they showed that SetDisjointness can be solved almost trivially in linear space and $O(\sqrt{N})$ query time, where $N$ is the total number of elements in all sets. This solution can be generalized to a full tradeoff between the space $S$ and the query time $T$ such that $S \cdot T^2 = O(N^2)$. For the SetIntersection problem, Cohen and Porat demonstrated a linear space solution with $O(N\sqrt{N})$ preprocessing time and $O(\sqrt{N}\sqrt{out}+out)$ query time, where $out$ is the output size. This was further generalized by Cohen~\cite{Cohen10} to a solution that uses $O(N^{2-2t})$ space with $O(N^{2-t})$ preprocessing time and $O(N^tout^{1-t}+out)$ query time for $0 \leq t \leq 1/2$.

From the lower bound prespective, \Patrascu{}~\cite{Patrascu10} proved the conditional \emph{time} hardness of the multiphase problem, which is a dynamic version of the SetDisjointness problem, based on the 3SUM conjecture. He also proved a connection between 3SUM and reporting triangles in a graph which is closely related to the SetIntersection problem. His conditional hardness results were improved by Kopelowitz et al.~\cite{KPP16} that considered the \emph{preprocessing and query time} tradeoff of both SetDisjointness and SetInteresection. Specifically, they proved, based on the 3SUM conjecture, that SetDisjointness has the following lower bound on the tradeoff between preprocessing time $T_p$ and query time $T_q$ for any $0 <\gamma<1$: $T_p+ N^{\frac{1+\gamma}{2-\gamma}}T_q = \Omega(N^{\frac{2}{2-\gamma}-o(1)})$.  Moreover, based on the 3SUM conjecture they also proved that SetIntersection has the following lower bound on the tradeoff between preprocessing, query and reporting (per output element) time for any $0 \leq\gamma<1, \delta>0$: $T_p+ N^{\frac{2(1+\gamma)}{3+\delta-\gamma}}T_q + N^{\frac{2(2+\delta)}{3+\delta-\gamma}}T_r= \Omega(N^{\frac{4}{3+\delta-\gamma}-o(1)})$. Kopelowitz et al.~\cite{KPP15} also proved the conditional \emph{time} hardness of the \emph{dynamic} versions of SetDisjointness and SetInteresection.

The lower bound on the \emph{space-query time} tradeoff for solving SetDisjointness was considered by Cohen and Porat~\cite{CP102} and \Patrascu{} and Roditty~\cite{PR14}. They have the following conjecture regarding the hardness of SetDisjointness (this is the formulation of Cohen and Porat. \Patrascu{} and Roditty use slightly different formulation):

\begin{conjecture}\label{conj:SDC}
\textbf{SetDisjointness Conjecture}. Any data structure for the SetDisjointness problem with constant query time must use $\tilde{\Omega}(N^2)$ space.
\end{conjecture}

Recently, Goldstein et al.~\cite{GKLP17} considered \emph{space} conditional hardness in a broader sense and demonstrated the conditional hardness of SetDisjointness and SetInteresection with regard to their space-query time tradeoff. They had a generalized form of Conjecture~\ref{conj:SDC} that claims that the whole (simple) space-time tradeoff upper-bound for SetDisjointness is tight:

\begin{conjecture}\label{conj:SSDC}
\textbf{Strong SetDisjointness Conjecture}. Any data structure for the SetDisjointness problem that answers queries in $T$ time must use $S=\tilde{\Omega}(\frac{N^2}{T^2})$ space.
\end{conjecture}

Moreover, they also presented a conjecture regarding the space-time tradeoff for SetIntersection:

\begin{conjecture}\label{conj:SSIC}
\textbf{Strong SetIntersection Conjecture}. Any data structure for the SetIntersection problem that answers queries in $O(T + out)$ time, where $out$ is the size of the output of the query, must use $S=\tilde{\Omega}(\frac{N^2}{T})$ space.
\end{conjecture}

Goldstein et al.~\cite{GKLP17} showed connections between these conjectures and other problems like 3SUM-Indexing (a data structure variant of 3SUM), k-Reachability and boolean matrix multiplication.
Unconditional lower bounds for the \emph{space-time} tradeoff of SetDisjointness and SetIntersection were proven by Dietz et al.~\cite{DMRU95} and Afshani and Nielsen~\cite{AN16} for specific models of computation. The results of Dietz et al.~\cite{DMRU95} implies that Conjecture~\ref{conj:SSDC} is true in the semi-group model. Afshani and Nielsen~\cite{AN16} proved Conjecture~\ref{conj:SSIC} in the pointer-machine model.

The fundamental nature of SetDisjointness and SetIntersection makes them useful for proving conditional lower bounds especially when considering their connection to the 3SUM problem. Indeed, several conditional lower bounds where proven using these problems (see~\cite{CP102,DSW12,GKLP16,KPP16,PR14,PRT12}). One major problem with this approach is that the universe of the elements in the sets of the SetDisjointness and SetIntersection problems can be large. This may cause the reduction from these problems to other problems, which we wish to prove their conditional hardness, to be inefficient. Therefore, it is of utmost interest to obtain a conditional lower bound on the hardness of SetDisjointness and SetIntersection with bounded universe, which in turn will be fruitful for achieving conditional lower bounds for other applications.

\textbf{Our Results}. In this paper we prove several conditional lower bounds for SetDisjointness and SetIntersection with bounded universe. We obtain the following results regarding the interplay between space and query time for solving these problems: (1) Based on the Strong SetDisjointness Conjecture, we prove that SetDisjointness with $m$ sets from universe $[u]$ must either use $\Omega(m^{2-o(1)})$ space or have $\Omega(u^{1/2-o(1)})$ query time. (2) Based on the Strong SetDisjointness Conjecture, we prove that SetIntersection with $m$ sets from universe $[u]$ must either use $\Omega(m^{2-o(1)})$ space or have $\tilde{\Omega}(u^{\alpha-o(1)} + out)$ query time, for any $1/2 \leq \alpha \leq 1$ and any output size $out$ such that $out=\Omega(u^{2\alpha-1-\delta})$ and $\delta>0$ (3) Based on the Strong SetIntersection Conjecture, we prove that SetIntersection with $m$ sets from universe $[u]$ must either use $\Omega((m^2u^{\alpha})^{1-o(1)})$ space or have $\tilde{\Omega}(u^{\alpha-o(1)} + out)$ query time for any $1/2 \leq \alpha \leq 1$ and any output size $out$ such that $out=\Omega(u^{2\alpha-1-\delta})$ and $\delta>0$.

Regarding the interplay of preprocessing and query time we demonstrate a reduction from 3SUM to SetDisjointness and SetIntersection. Using this reduction we prove the following results based on the 3SUM conjecture:
(i) Any solution to SetDisjointness with $m$ sets from universe $[u]$ must either have $\Omega(m^{2-o(1)})$ preprocessing time or have $\Omega(u^{1/2-o(1)})$ query time. (ii) Any solution to SetIntersection with $m$ sets from universe $[u]$ must either have $\Omega(m^{2-o(1)})$ preprocessing time or have $\Omega(u^{1-o(1)})$ query time.

These new conditional lower bounds are useful in proving conditional lower bounds for other problems that exploit the small universe size as explained before. We give some examples of such applications.

\smallskip

(1) \textbf{Range Mode}. The Range Mode problem is a classic problem that was studied in several papers (see e.g.~\cite{CDLMW14,KMS05}). In this problem, an array $A$ with $n$ elements is given for preprocessing. Then, we are required to answer range mode queries. That is, given a range $[i,j]$ we have to find the mode element (the most frequent element) in the range $[i,j]$ in $A$. The best known upper bound for the space-query time tradeoff of this problem is $S \cdot T^2 = \tilde{O}(n^2)$, where $S$ is the space usage and $T$ is the query time (\cite{CDLMW14,KMS05}). We prove using our new lower bound for SetDisjointess with bounded universe the following lower bound on the tradeoff between space and query time: $S \cdot T^4 = \Omega(n^{2-o(1)})$. We note that if the query time in the lower bound on SetDisjointness (in Theorem~\ref{thm:SD_small_universe_LB}, see (1) above) was $\tilde{\Omega}(u^{1-o(1)})$ then the lower and upper bounds were tight.

\smallskip

(2) \textbf{Distance oracle} is a data structure for computing the shortest path between any two vertices in a graph. We say that a distance oracle has a stretch $t$ if for any two vertices in the graph the distance it returns is no more than $t$ times the true distance between these vertices. Approximate distance oracles were investigated in many papers (see for example~\cite{AG13,Agarwal14,Chechik14,PR14,PRT12,TZ05}). Agrawal~\cite{Agarwal14} showed a $(\frac{5}{3})$-stretch distance oracle for a graph $G=(V,E)$ that uses $\tilde{O}(|E|+\frac{|V|^2}{\alpha})$ space and has $O(\alpha\frac{|E|}{|V|})$ query time for any $1 \leq \alpha \leq (\frac{|V|^2}{|E|})^{\frac{1}{3}}$. We prove that this tradeoff is the best that can be achieved for any stretch-less-than-2 distance oracles based on our new lower bound on SetDisjointness with bounded universe (see a more detailed discussion in Section~\ref{sec:applications}).

\smallskip

(3) \textbf{3SUM-Indexing} is a data structure variant of 3SUM. In this problem, two arrays $A$ and $B$ with $n$ numbers in each of them are preprocessed. Then, given a query number $z$ we are required to decide if there are $x \in A$ and $y \in B$ such that $x+y=z$. Goldstein et al.~\cite{GKLP17} conjecture that there is no $\tilde{O}(1)$ query time solution to 3SUM-Indexing using truly subquadratic space. In a stronger form of this conjecture they claim that there is no truly sublinear query time solution to 3SUM-Indexing using truly subquadratic space. Goldstein et al.~\cite{GKLP17} proved some connections between 3SUM-Indexing, SetDisjointness and SetIntersection. In this paper we strengthen their results using our new lower bounds for SetDisjointness and SetIntersection with bounded universe. Specifically, we prove based on our new lower bound on SetDisjointness with bounded universe that any solution to 3SUM-Indexing where the universe of the numbers within arrays $A$ and $B$ is $[n^{2+\epsilon}]$
for any $\epsilon>0$ has this lower bound on the tradeoff between space ($S$) and query time ($T$): $S \cdot T^2 = \Omega(n^{2-o(1)})$. Moreover, we prove the same lower bound on the tradeoff between preprocessing ($T_p$) and query time ($T_q$): $T_p \cdot T_q^2 = \tilde{\Omega}(n^{2-o(1)})$. The latter is proven based on the 3SUM conjecture following our reduction to SetDisjointness with bounded universe.

In the 3SUM conjecture the universe of the numbers in the given instance is assumed to be $[n^3]$ (see~\cite{Patrascu10}) or even $[n^4]$ (see~\cite{Williams18}). It is known that 3SUM can be easily solved in $O(u\log{u})$ time if the universe is $[u]$ by using FFT. Therefore, 3SUM with numbers from universe $[n^{2-\epsilon}]$ for any $\epsilon>0$ can be solved in truly subquadratic time. Consequently, assuming that no truly subquadratic solution to 3SUM with universe $[n^2]$ seems to be much stronger conjecture (it was used once in~\cite{ACLL14}). Solving 3SUM-Indexing can be done easily with $\tilde{O}(n^2)$ preprocessing time, $O(n^2)$ space and $\tilde{O}(1)$ query time. Our results demonstrate that this is tight even if the universe of the numbers in $A$ and $B$ and the query numbers is $[n^{2+\epsilon}]$ for any $\epsilon>0$. This is a very strong lower bound, as 3SUM-Indexing with numbers from universe $[u]$ can be solved with $\tilde{O}(u)$ preprocessing time, $O(u)$ space and $\tilde{O}(1)$ query time. This is done in a similar way to solving 3SUM with numbers from universe $[u]$. Consequently, for any $\epsilon>0$, 3SUM-Indexing with numbers from universe $[n^{2-\epsilon}]$ can be solved by a data structure that has constant query, while the preprocessing time and space are subquadratic. Our new conditional lower bound demonstrates that having such a data structure for a slightly larger universe seems to be impossible.

\section{Hardness of Space-Time Tradeoff for SD and SI with Bounded Universe}
We prove the hardness of SetDisjointness with bounded universe in the following theorem:
\begin{theorem}\label{thm:SD_small_universe_LB}
Any solution to SetDisjointness with sets $S_1,S_2,...,S_m \subseteq [u]$ for any value of $u \in [N^{\delta},N]$, such that $N = \sum_{i=1}^{m}{|S_i|}$ and $\delta>0$, must either use $\Omega(m^{2-o(1)})$ space or have $\tilde{\Omega}(u^{1/2-o(1)})$ query time, unless the Strong SetDisjointness Conjecture is false.
\end{theorem}
\begin{proof}
Let us assume to the contradiction that the Strong SetDisjointness Conjecture is true, but there is an algorithm $A$ that solves SetDisjointness on $m$ sets from a universe $[u]$ and creates a data structure $D$, such that the space complexity of the data structure $D$ is $O(m^{2-\epsilon_1})$ for some $\epsilon_1 > 0$ and the query time of algorithm $A$ is $O(u^{1/2-\epsilon_2})$ for some $0<\epsilon_2 \leq 1/2 $. We define $\epsilon=\min(\epsilon_1,\epsilon_2)$.

Now, given an instance of SetDisjointness with sets $S'_1,S'_2,...,S'_{m'}$, we denote by $N'$ the total number of elements in all sets, that is $N' = \sum_{i=1}^{m}{|S'_i|}$. We rename the elements of all the sets such that each element $e_i$ is mapped to some integer $x_i \in [N']$.

We distinguish between 3 types of sets:
\begin{enumerate}[label=(\alph*)]
\item \textbf{Large sets} are all the sets with more than $\sqrt{u}$ elements. Denote by $d$ the number of large sets. Let $S_{p_1}, S_{p_2},...,S_{p_d}$ be some ordering of the large sets. Let $p$ be a function such that $p(i)=p_j$ if $S_i$ is the set $S_{p_j}$ in the ordering of the large sets.
\item \textbf{Small sets} are all the sets with $O(u^{1/2-\epsilon})$ elements.
\item \textbf{Medium sets} are all the sets that are neither large nor small. Denote by $e$ the number of medium sets. Let $S_{q_1}, S_{q_2},...,S_{q_e}$ be some ordering of the medium sets. Let $q$ be a function such that $q(i)=q_j$ if $S_i$ is the set $S_{q_j}$ in the ordering of the medium sets.
\end{enumerate}

Now, we can solve SetDisjointness in the following way.

\medskip
\noindent
\textbf{Preprocessing}:
\begin{enumerate}[label=(\arabic*)]
 \item For any set $S_i$ use static hashing to save all elements of the set in a table $T_i$, such that we can check if some element exists in the set in $O(1)$ time and the size of $T_i$ is $O(|S_i|)$.
 \item Maintain a $d \times (d+e)$ matrix $M$. The $\ell$th row in this matrix represents the set $S_{p_{\ell}}$. For $1 \leq \ell \leq d$, the $\ell$th column represents $S_{p_{\ell}}$ and for $d+1 \leq \ell \leq d+e$, the $\ell$th column represents $S_{q_{\ell - d}}$.
 \item For all pairs of sets $S_i$ and $S_j$ such that $S_i$ is a large set and $S_j$ is a large or medium set, save an explicit answer to the emptiness of the intersection of $S_i$ and $S_j$ in $M[p(i),p(j)]$ and $M[p(j),p(i)]$ if $S_j$ is a large set and in $M[p(i),d+q(j)]$ if $S_j$ is a medium set.
 \item Pick $\log{n}$ hash functions $h_i:N \rightarrow [8u]$, for $1 \leq i \leq \log{n}$. Apply each $h_i$ to all elements in all medium sets. Denote by $h_i(S_j)$ the set $S_j$ after $h_i$ has been applied to its elements.
 \item For every $i,j \in [e]$, if $S_{q_{i}} \cap S_{q_{j}} = \emptyset$ do the following: Check if for all $k \in [\log{n}]$ there are $x \in S_{q_{i}}$ and $x' \in S_{q_{j}}$ such that $x \neq x'$ but $h_k(x)=h_k(x')$. If so, go back to step (4).
 \item For every $k \in [\log{n}]$:
  \begin{enumerate}[label=(6.\arabic*)]
  \item Apply $h_k$ to all the elements of all the medium sets.
  \item Use algorithm $A$ to create a data structure $D_k$ that solves the set disjointness problem on the medium sets $S_{q_1}, S_{q_2},...,S_{q_e}$ after $h_k$ has been applied to their elements.
    \end{enumerate}
\end{enumerate}

\noindent
\textbf{Query}:

Given a pair of indices $i$ and $j$, we need to determine if $S_i \cap S_j$ is empty or not. Without loss of generality we assume that $|S_i| < |S_j|$ and do the following:

\begin{enumerate}[label=(\arabic*)]
\item If $S_i$ is a small set:
     \begin{enumerate}[label=(1.\arabic*)]
     \item For each element $x \in S_i$:
            Check if $x \in S_j$ using table $T_j$.
            If so, return 0.
     \item Return 1.
     \end{enumerate}
\item If $S_j$ is a large set:
      \begin{enumerate}[label=(2.\arabic*)]
      \item If $S_i$ is a large set: Return $M[p(i),p(j)]$.
      \item If $S_i$ is a medium set: Return and $M[p(j),d+q(i)]$.
      \end{enumerate}
\item Else (if both $S_i$ and $S_j$ are medium sets):

\setlength{\leftskip}{6pt} For every $k \in [\log{n}]$, check by using algorithm $A$ and the data structure $D_k$ if $S_i$ and $S_j$ are disjoint.

If there is at least one value of $k$ for which these sets are disjoint, return 1.

Otherwise, return 0.
\end{enumerate}

\textbf{Correctness.} If at least one of the query sets is small then we can check if any of its elements is in the other query set using the hash tables that have been created in step (1) of the preprocessing phase. This is done in step (1) of the query algorithm. If at least one of the sets is large we can find the answer immediately by looking at the right position of matrix $M$ that has been created in steps (2)-(3) of the preprocessing phase. The last option is that both query sets are medium. If this is the case we use the data structures that have been created in step (6) of the preprocessing phase. In steps (4) and (5) of the preprocessing phase we look for $\log{n}$ hash functions such that if any pair of sets are disjoint then they must be disjoint when applying the hash functions to their elements by at least one of the $\log{n}$ hash functions. Therefore, if any of the data structures that have been created in step (6) of the preprocessing phase reports that a pair of sets are disjoint they must be disjoint. Moreover, if a pair of sets are disjoint then there must be at least one data structure that reports that they are disjoint. This is checked in the step (3) of the query algorithm.

The last thing that needs to be justified is the existence of $\log{n}$ hash function such that for every pair of sets $S_i$ and $S_j$ that are disjoint they are also disjoint after applying the hash functions by at least one of the $\log{n}$ hash functions. The range of the hash function is $[8u]$. The number of elements in the medium sets is no more than $\sqrt{u}$. Therefore, for any two medium sets $S_i$ and $S_j$ and a hash function $h_k:N \rightarrow [8u]$ we have by the union-bound that $\Pr[\exists x_1 \in S_i, x_2 \in S_j: x_1 \neq x_2 \wedge h_k(x_1)=h_k(x_2)] \leq \frac{\sqrt{u} \cdot \sqrt{u}}{8u} = 1/8$. Consequently, the probability that a pair of disjoint medium sets $S_i$ and $S_j$ are not disjoint when applying $h_k$ for all $k \in [\log{n}]$ is no more than $(1/8)^{\log{n}} = 1/n^3$. Therefore, the probability that any pair of disjoint medium sets are not disjoint when applying $h_k$ for all $k \in [\log{n}]$ is no more than $n^2/n^3 = 1/n$ by the union-bound. Using the probabilistic method we get that there must be $\log{n}$ hash functions such that for every pair of sets $S_i$ and $S_j$ that are disjoint they are also disjoint after applying the hash functions by at least one of the $\log{n}$ hash functions.

\medskip

\textbf{Complexity analysis}.

\textbf{Space complexity}. The space for the tables in step (1) of the preprocessing is clearly $O(N)$ - linear in the total number of elements. The total number of large sets $d$ is at most $O(N/u^{1/2})$. The total number of medium sets $e$ is at most $O(N/u^{1/2-\epsilon})$. Therefore, the size of the matrix $M$ is at most $O(N/u^{1/2} \cdot (N/u^{1/2} + N/u^{1/2-\epsilon})) = O(N^2/u^{1-\epsilon})$. There are $\log{n}$ data structures that are created in step (6). Each data structure uses at most $O((N/u^{1/2-\epsilon})^{2-\epsilon}) = O(N^{2-\epsilon}/u^{1-5\epsilon/2+\epsilon^2})$ space. Consequently, the total space complexity is $S=\tilde{O}(N^2/u^{1-\epsilon}+N^{2-\epsilon}/u^{1-5\epsilon/2+\epsilon^2})$.

\textbf{Query time complexity}. Step (1) of the query algorithm can be done in $O(u^{1/2-\epsilon})$ as this is the size of the largest small set. Step (2) is done in constant time by looking at the right position in $M$. In step (3) we do $\log{n}$ queries using algorithm $A$ and the data structures $D_k$. The query time for each query is $O(u^{1/2-\epsilon})$ as the universe of the sets after applying any hash function $h_k$ is $[8u]$. Therefore, the total query time is $T=O(u^{1/2-\epsilon})$.

\smallskip

Following our analysis we have that $S \cdot T^2 = \tilde{O}((N^2/u^{1-\epsilon}+N^{2-\epsilon}/u^{1-5\epsilon/2+\epsilon^2}) \cdot (u^{1/2-\epsilon})^2) = \tilde{O}(N^2u^{-\epsilon}+N^{2-\epsilon}u^{\epsilon/2-\epsilon^2})=\tilde{O}(N^2u^{-\epsilon}+N^2u^{-\epsilon^2})$ (the last equality follows from the fact that $u \leq N$).
This contradicts the Strong SetDisjointness Conjecture and therefore our assumption is false.
\end{proof}

From the proof of the above theorem we get a specific range for the value of $m$ for hard instances of SetDisjointness. Bounding the value of $m$ for hard instances may be useful for some specific applications. Therefore, we state the following corollary of the proof of Theorem~\ref{thm:SD_small_universe_LB}:

\begin{corollary}\label{thm:SD_small_universe_LB_for_specific_m}
For any $\epsilon>0$, any solution to set disjointness with sets $S_1,S_2,...,S_m \subseteq [u]$ for any value of $u \in [N^{\delta},N]$, such that $N = \sum_{i=1}^{m}{|S_i|}$, $\delta>0$ and the solution works for any value of $m$ in the range $[\frac{N}{u^{1/2}},\frac{N}{u^{1/2-\epsilon}}]$, must either use $\Omega(m^{2-o(1)})$ space or have $\Omega(u^{1/2-o(1)})$ query time, unless the Strong SetDisjointness Conjecture is false.
\end{corollary}

We also prove conditional lower bounds on SetIntersection with bounded universe based on the Strong SetDisjointness Conjecture and the Strong SetIntersection Conjecture by generalizing the ideas from the previous proof. These results appear in Appendix~\ref{sec:clb_setintersection}.

\section{Hardness of Preprocessing-Query Time Tradeoff for SD and SI with Bounded Universe}\label{sec:3sum_clb_sd_si}
We combine the ideas of Goldstein et al.~\cite{GKLP16} and Kopelowitz et al.~\cite{KPP16} to get conditional lower bounds on the complexity of SetDisjointness with bounded universe. To achieve these bounds we prove the following lemma:

\begin{lemma}\label{lem:3sum_indexing_setdisjointness_reduction}
Let $X$ be any integer in $[n^{\delta},n]$ for any $\delta>0$. For any $\epsilon>0$, an instance of 3SUM-Indexing that contains 2 arrays with $n$ integers can be reduced to $2\epsilon\log{X}$ instances of SetDisjointness $SD_1,SD_2,...,SD_{2\epsilon\log{X}}$. For any $1 \leq i \leq 2\epsilon\log{X}$, instance $SD_i$ have $N_i=n\sqrt{u_i}$ elements from universe $[u_i]$ and $m=n\sqrt{\frac{X}{u_i}}$ sets that each one of them is of size $O(\sqrt{u_i})$, where $u_i=X^{1+\epsilon}/2^{i-1}$. The time and space complexity of the reduction is truly subquadratic in $n$. Each query to the 3SUM-Indexing instance can be answered by at most $O(n/\sqrt{X})$ queries to each instance $SD_i$ plus some additional time that is truly sublinear in $n$.
\end{lemma}

\begin{proof}
We begin with an instance of $3$SUM indexing with arrays $A$ and $B$ and do the following construction in order to reduce this $3$SUM indexing instance to $2\epsilon\log{n}$ instances of SetDisjointness. The construction uses almost-linear and almost-balanced hash functions that serve as a useful tool in many reductions from 3SUM. We briefly define this notion here (see full details in~\cite{KPP16,Wang14}). Let $\mathcal{H}$ be a family of hash functions from $[u] \rightarrow [m]$. $\mathcal{H}$ is called {\em linear} if for any $h\in\mathcal{H}$ and any $x,x' \in [u]$, we have $h(x) + h(x') \equiv h(x+x') \; (\modulo m)$.
$\mathcal{H}$ is called {\em almost-linear} if for any $h\in\mathcal{H}$ and any $x,x' \in [u]$, we have
either $h(x) + h(x') \equiv h(x+x') +c_h \; (\modulo m)$, or $h(x) + h(x') \equiv h(x+x') + c_h +1 \; (\modulo m)$, where $c_h$ is an integer that depends only on the choice of $h$.
For a function $h:[u] \rightarrow [m]$ and a set $S\subset [u]$ where $|S|=n$, we say that $i\in [m]$ is an overflowed value of $h$ if $|\{x\in S : h(x) = i\}| > 3n/m$.
$\mathcal{H}$ is called {\em almost-balanced} if for a random $h\in \mathcal{H}$ and any set $S\subset [u]$ where $|S|=n$, the expected number of elements from $S$ that are mapped to overflowed values is $O(m)$. For simplicity of presentation, we treat the almost-linear hash functions as linear and this only affects some constant factors in our analysis.

\textbf{Construction}.

\textbf{Initial Construction}. We use an almost-linear almost-balanced hash function $h_1: U \rightarrow [R]$ to map the elements of $A$ to $R$ buckets $A_1,A_2,...,A_R$ such that $A_i = \{x \in A: h_1(x)=i\}$ and the elements of $B$ to $R$ buckets $B_1,B_2,...,B_R$ such that $B_i = \{x \in B: h_1(x)=i\}$. As $h_1$ is almost-balanced the expected size of each bucket is $O(n/R)$. Moreover, buckets with more than $3n/R$ elements, called overflowed buckets, have no more than $O(R)$ elements in total. We save these $O(R)$ elements in lists $L_A$ and $L_B$ (we put elements from overflowed buckets of $A$ in $L_A$ and elements from  overflowed buckets of $B$ in $L_B$). We also sort $A$ and $B$ and save lookup tables for both $A$ and $B$.

We pick another almost-linear almost-balanced hash function $h_2: U \rightarrow [n]$.
For each bucket $A_i$, we create an $n$-length characteristic vector $v_{A_i}$ such that $v_{A_i}[j]=1$ if there is $x \in A_i$ such that $h_2(x)=j$ and $v_{A_i}[j]=0$ if there is no $x \in A_i$ such that $h_2(x)=j$. In the same way we create an $n$-length characteristic vector $v_{B_j}$ for each bucket $B_j$.

\textbf{Quad Trees Construction}. We create a search quad tree for each pair of buckets $A_i$ and $B_j$ following the idea of Goldstein et al.~\cite{GKLP16}. The construction involves calculating the \emph{convolution} of many pairs of vectors. The \emph{convolution} of two vectors $u,v\in \{\mathbb{R}^+\cup\{0\}\}^n$ is a vector $c$, such that $c[k]=\sum_{i=0}^{k}{u[i]v[k-i]}$ for $0 \leq k \leq 2n-2$. Constructing the quad tree is done as follows:

\smallskip

\noindent
\textbf{Quad-Tree-Construction($v_{A_i}$,$v_{B_j}$,$X$)}
\begin{enumerate}[label=(\arabic*)]
  \item For the bottom level of the quad tree:
      \begin{enumerate}[label=(1.\arabic*)]
      \item Partition the characteristic vector $v_{A_i}$ into $\lceil n/X \rceil$ sub-vectors $v_{{A_i}_1},...,v_{{A_i}_{\lceil n/X \rceil}}$ each of them of length $X$.
      \item Pad the last sub-vector with zeroes if needed.
      \item Let $i_1,i_2,...,i_Y$ be the indices of the ones in some sub-vector $v_{{A_i}_k}$.
           If $Y>X/R$
           \begin{enumerate}[label=(1.3.\arabic*)]
             \item Duplicate $v_{{A_i}_k}$ $t=\lceil Y/(X/R) \rceil$ times.
             \item For every $p \in [t]$:
                Save in the $p$th copy of $v_{{A_i}_k}$ just the ones in the indices $i_{(p-1)\cdot(X/R)+1},...,i_{p\cdot(X/R)-1}$.
                Replace all other ones by zeroes.
           \end{enumerate}
      \item Denote the sequence of sub-vectors of $v_{A_i}$ and their duplicates by \\ $P_{A_i}= v_{A_i}^1, v_{A_i}^2,..., v_{A_i}^{cn/X}$ for some constant $c \geq 1$. Order the sub-vectors in $P_{A_i}$ by the locations of the ones. That is, sub-vector $w$ occurs before $u$ in $P_{A_i}$ if the ones in $w$ appear before the ones of $u$ in $v_{A_i}$. A sub-vector $w$ that contains only zeroes and therefore represents a sub-vector $v_{{A_i}_k}$ for some $1<k\leq \lceil n/X \rceil$ without any duplicates appears before all sub-vectors $v_{{A_i}_{k'}}$ for $k'>k$ and their duplicates.
      \item Repeat steps (1.1)-(1.4) for $v_{B_j}$ and create a sequence of sub-vectors \\ $P_{B_j}=v_{B_j}^1, v_{B_j}^2,..., v_{B_j}^{c'n/X}$ for some constant $c' \geq 1$.
      \item Without loss of generality let us assume that $c \geq c'$. Add to the end of the sequence $P_{B_j}$ the vectors $v_{B_j}^{c'n/X+1},...,v_{B_j}^{cn/X}$, such that each of these vectors contains exactly $X$ zeroes.
      \item For each pair of sub-vectors $v_{A_i}^k$ and $v_{B_j}^{\ell}$:
                \begin{enumerate}[label=(1.7.\arabic*)]
                \item Create a node $c_{i,j}^{k,\ell}$ in the quad tree.
                \item Calculate the convolution of $ v_{A_i}^k$ and $v_{B_j}^{\ell}$ and save the result in $c_{i,j}^{k,\ell}$.
                \end{enumerate}
      \end{enumerate}

  \item For the next level of the quad tree upward:
    \begin{enumerate}[label=(2.\arabic*)]
    \item Create a sequence of sub-vectors ${v'}_{A_i}^1, {v'}_{A_i}^2,..., {v'}_{A_i}^{cn/2X}$ such that ${v'}_{A_i}^k$ is the concatenation of $v_{A_i}^{2k-1}$ and $v_{A_i}^{2k}$ from the previous level.
    \item For every ${v'}_{A_i}^k$ if there are overlapping locations in $v_{A_i}^{2k-1}$ and $v_{A_i}^{2k}$ - merge them. That is, if there are elements in both sub-vectors that represent the same interval of $v_{A_i}$, merge all of them in ${v'}_{A_i}^k$ by setting each overlapping location to 1 if any of the two overlapping elements in this location is 1,  and setting each overlapping location to 0 otherwise.
    \item Repeat steps (2.1) and (2.2) for $v_{B_j}$ and create a sequence ${v'}_{B_j}^1, {v'}_{B_j}^2,..., {v'}_{B_j}^{cn/2X}$.
    \item For each pair of sub-vectors ${v'}_{A_i}^k$ and ${v'}_{B_j}^{\ell}$ create a node ${c'}_{i,j}^{k,\ell}$ in the quad tree.
    \item Make the node ${c'}_{i,j}^{k,\ell}$ the parent of 4 nodes from the previous level: \\ $c_{i,j}^{2k-1,\ell-1},c_{i,j}^{2k-1,\ell},c_{i,j}^{2k,\ell-1},c_{i,j}^{2k,2\ell}$.
    \item Calculate the convolution of $ {v'}_{A_i}^k$ and ${v'}_{B_j}^{\ell}$ and save the result in ${c'}_{i,j}^{k,\ell}$. The convolution of ${v'}_{A_i}^k$ and ${v'}_{B_j}^{\ell}$ can be easily calculated using the convolution results that are saved in $c_{i,j}^{2k-1,\ell-1},c_{i,j}^{2k-1,\ell},c_{i,j}^{2k,\ell-1},c_{i,j}^{2k,2\ell}$ from the previous level.
    \end{enumerate}

  \item Repeat step (2) for all the levels up to the root. Notice that in the root we have the complete vectors $v_{A_i}$ and $v_{B_j}$ and we calculate and save their convolution within the root node.
\end{enumerate}

We emphasize that in the bottom level of the quad tree the number of sub-vectors of $v_{A_i}$ including all duplicates is no more than $cn/X$ for some constant $c \geq 1$, as the total number of ones in $v_{A_i}$ is $O(n/R)$. Therefore, the size of the sequence in step (1.4) is $cn/X$.

We call a quad tree such that the length of the sub-vectors in its bottom level is $X$ $X$-quad-tree. We denote the level of the quad tree with sub-vectors of length $Z$ by $\ell_Z$. We emphasize that we consider the length of the sub-vectors for the last notation by their length if we do no merging in any level of the quad tree.

\textbf{Convolution by SetDisjointness}. The convolution $c$ of two $X$-length vectors $v$ and $u$ can be calculated using SetDisjointness in the following way: Let us denote by $v_i$ (for any $0 \leq i \leq X-1$) a $(2X-1)$-length vector, such that $v_i[j+i]=v[j]$ for every $0 \leq j \leq X-1$ and all other elements of $v_i$ are zeroes. It is clear that $v_i$ is the vector $v$ that its elements where shifted by $i$ locations and the empty locations are filled with zeroes. Therefore, we call the vector $v_i$ an $i$-shift of $v$. We define $u_i$ in a similar way. Let us denote by $v^R$ the vector $v$ in reverse order of elements. It is straightforward to observe that $c[j]$ (the $j$th element in the convolution result of $v$ and $u$) equals to the inner product of $v^R_j$ (we note that the reverse operation is done before the shift operation) and $u_{X-1}$. Informally, the complete convolution of $v$ and $u$ can be calculated by the inner product of (padded) $u$ and the reversed version of (padded) $v$ in $X-1$ different shifts. We can reduce the number of shifts to $v$ by shifting both $v$ and $u$. Specifically, the value of $c[j]$ can be obtained by the inner product of $v^R_{j \mod \sqrt{X}}$ and $u_{X-1-\lfloor \frac{j}{\sqrt{X}} \rfloor \cdot \sqrt{X}}$. Therefore, the convolution of $v$ and $u$ can be calculated by the inner product of $O(\sqrt{X})$ shifted versions of both $v$ and $u$.

Each of the $(2X-1)$-length boolean vectors can be represented by a set corresponding to the ones in the vector. Formally, for a vector $w$ we construct a set $S_{w}$ such that $S_{w} = \{j | w[j] = 1\}$. Instead of calculating the inner product of $v^R_{j \mod \sqrt{X}}$ and $u_{X-1-\lfloor \frac{j}{\sqrt{X}} \rfloor \cdot \sqrt{X}}$, we can calculate $|S_{v^R_{j \mod \sqrt{X}}} \cap S_{u_{X-1-\lfloor \frac{j}{\sqrt{X}} \rfloor \cdot \sqrt{X}}}|$ and get the same result. In our query process through the quad tree we just need to know in each node if the value in some position of the convolution within that node is zero or not. Thus, instead of calculating $|S_{v^R_{j \mod \sqrt{X}}} \cap S_{u_{X-1-\lfloor \frac{j}{\sqrt{X}} \rfloor \cdot \sqrt{X}}}|$ we just need to determine if $S_{v^R_{j \mod \sqrt{X}}} \cap S_{u_{X-1-\lfloor \frac{j}{\sqrt{X}} \rfloor \cdot \sqrt{X}}} = \emptyset$ or not. All in all, the convolution of two $X$-length vectors $v$ and $u$ can be determined by a SetDisjointness instance that contains $O(\sqrt{X})$ sets such that their size equals to the number of ones in either $v$ or $u$.
Consequently, instead of saving explicitly the convolution result in each node in some level of the quad tree that represents sub-vectors of length $X$, we can create an instance of SetDisjointness that can be used to determine if a specific position in a convolution result is zero or not.

\textbf{Hybrid Quad Tree Construction}. Using the idea from the previous paragraph we modify the quad tree construction in the following way: We construct in the regular way, that is explained in detail above, each of the quad trees until level $\ell_{X^{1-\epsilon}}$. From level $\ell_{X^{1-\epsilon}}$ to level $\ell_{X^{1+\epsilon}}$ we do not save the convolution results explicitly in the quad tree for each level, but rather we create a SetDisjointness instance that can be used to answer if a specific position in a convolution result is zero or not. This is an hybrid construction in which we create an $(X^{1+\epsilon})$-quad-tree that the bottom $X^{2\epsilon}$ levels are not saved explicitly. Instead, the information for these bottom levels is determined by the SetDisjointness instances we create. These levels are called the implicit levels of the hybrid quad tree while the levels in which we save the convolution results explicitly are called the explicit levels of the hybrid quad tree.

\smallskip

\textbf{Query}.

Given a query integer number $z$, we search for a pair of integers $x \in A$ and $y \in B$ such that $x+y=z$. First of all, we check for each element $x \in L_A$ if there is $y \in B$ such that $x+y=z$ and we also check for each element $y \in L_B$ if there is $x \in A$ such that $x+y=z$. This can be done easily in $\tilde{O}(R)$ time using the sorted versions of $A$ and $B$. Then, if $x$ is in bucket $A_i$ then by the (almost) linearity property of $h_1$ we expect $y$ to be in bucket $B_j$ such that $j=i-h_1(z)$.
In order to find out if there is $x \in A_i$ and $y \in B_j$ such that $x+y=z$ we can calculate the convolution of $v_{A_i}$ and $v_{B_j}$. Denote the vector that contains their convolution result by $C_{i,j}$. If $C_{i,j}[h_2(z)] = 0$ then there are no $x \in A_i$ and $y \in B_j$ such that $x+y=z$. However, if $C_{i,j}[h_2(z)] \neq 0$ then there may be $x \in A_i$ and $y \in B_j$ such that $x+y=z$, but it may also be the case that $h_2(x)+h_2(y)=h_2(z)$ while $x+y \neq z$. Therefore, in order to verify if there are $x \in A_i$ and $y \in B_j$ such that $x+y=z$, we need to find all pairs of $x' \in A_i$ and $y' \in B_j$ such that $h_2(x')+h_2(y')=h_2(z)$ and check if indeed $x'+y'=z$. There are exactly $C_{i,j}[h_2(z)]$ such pairs, which are also called witnesses.

In order to efficiently find the witnesses of $C_{i,j}[h_2(z)]$, we use the hybrid quad tree we have constructed for buckets $A_i$ and $B_j$ in the following way: We start at the root of the hybrid quad tree if the convolution result in the root is non-zero at location $h_2(z)$, we look at the children of the root node and continue the search at each child that contains a non-zero value in the convolution result it saves in the index that corresponds to index $h_2(z)$ of the convolution in the root. This way we continue downward all the way to the leaves. In the levels of the hybrid quad tree that the convolution results are not saved explicitly we query the SetDisjointness instances in order to get an indication for the existence of a witness in the search path from the root.

If we reach a leaf of the quad tree and the convolution result within this leaf is non-zero in the location that corresponds to the index $h_2(z)$ of the convolution in the root, then we do a "2SUM-like" search within this leaf.

The "2SUM-like" search is done as follows: Let us assume that the leaf represents 2 sub-vectors $v_{A_i}^k$ and $v_{B_j}^{\ell}$. We recover the original elements that these sub-vectors represent. Let the array $A_i^k$ contain all $x \in A_i$ such that there is one in $v_{A_i}^k$ that corresponds to $h_2(x)$. In the same way we construct array $B_j^{\ell}$. We sort both $A_i^k$ and $B_j^{\ell}$. Let $d$ be the size of $A_i^k$. Then, if $A_i^k[d-1]+B_j^{\ell}[0]=z$ we are done. Otherwise, if the sum is greater than $z$ we check if $A_i^k[d-2]+B_j^{\ell}[0]=z$ and if it is smaller than $z$ we check if $A_i^k[d-1]+B_j^{\ell}[1]=z$. This way we continue until we get to the end of one of the arrays or find a pair of elements that its sum equals $z$.

\smallskip

\textbf{Analysis}. There are $R^2$ possible pairs of buckets $A_i$ and $B_j$. Therefore, we  construct $R^2$ quad trees. In order to save the convolution results in all the nodes in an explicit level $\ell_Z$ of some hybrid quad tree, the space we need to use is $O(n^2/Z)$ (for each pair $A_i$ and $B_j$, there are $O(n^2/Z^2)$ pairs of sub-vectors one from $v_{A_i}$ and the other from $v_{B_j}$. The size of the convolution of the two sub-vectors is $O(Z)$). Therefore, the total space for constructing the explicit levels of the hybrid quad trees is $\tilde{O}(n^2/X^{1+\epsilon} \cdot R^2)$ (a level that is closer to the root requires less space than a level that is farther away from the root. There are at most $\log{n}$ levels in each quad tree. The bottom explicit level is $\ell_{X^{1+\epsilon}}$). This is also the preprocessing time for constructing these levels of the hybrid quad trees as the convolution of two $n$-length vector can be calculated in $\tilde{O}(n)$ time.

From level $\ell_{X^{1-\epsilon}}$ to level $\ell_{X^{1+\epsilon}}$ we do not save the convolution results explicitly in the quad tree for each level, but rather we create a SetDisjointness instance that can be used to answer if a specific position in a convolution result is zero or not, as explained in detail previously. Let us analyse the cost of the SetDisjointness instance for some implicit level $\ell_Z$. We have $O(R)$ buckets. Each bucket is represented by a characteristic vector that is partitioned into $O(n/Z)$ parts of length $Z$, such that each part contains $O(Z/R)$ ones. For each sub-vector we create $O(\sqrt{Z})$ sets that represent $O(\sqrt{Z})$ shifts of the sub-vector as explained previously. Therefore, the total number of sets we have is $O(R \cdot n/Z \cdot \sqrt{Z})=O(nR/\sqrt{Z})$. Each set contains $O(Z/R)$ elements, so the total number of elements in all sets is $O(R \cdot n/\sqrt{Z} \cdot Z/R)=O(n\sqrt{Z})$. The universe of all the elements in the sets is $Z$.

For a query integer $z$ we have $O(R)$ pairs of buckets $A_i$ and $B_j$ in which we may have two elements, one from each array, that sum up to $z$ (as $j=i-h_1(z)$). For a pair of buckets $A_i$ and $B_j$, we search for all the witnesses of $C_{i,j}[h_2(z)]$ in the quad tree of $A_i$ and $B_j$. Searching for a witness from the root to a leaf of the quad tree can be done in $O(\log{n})$ time in the levels we save the convolution explicitly and a constant number of queries for each SetDisjointness instance. Within a leaf we do a "2SUM-like" search on 2 arrays that contain $O(X^{1-\epsilon}/R)$ elements. Therefore, the total search time per witness is at most $\tilde{O}(X^{1-\epsilon}/R)$.
A false witness is a witness pair of elements $(x,y)$ such that $x+y \neq z$, but $h_2(x)+h_2(y) = h_2(z)$. The probability that a pair of numbers $(x,y)$ is a false witness is $1/n$ (because the range of $h_2$ is $[n]$). Therefore, the expected number of false witnesses within a specific pair of buckets is at most $O((n/R)^2\cdot 1/n) =O(n/R^2)$ by the union-bound (notice that the number of elements in each bucket is $O(n/R)$). Consequently, the total expected number of false witnesses is at most $O(Rn/R^2)=O(n/R)$.
As explained before, the total search time per witness is at most $\tilde{O}(X^{1-\epsilon}/R)$. Thus, the total query time is $\tilde{O}(nX^{1-\epsilon}/R^2)$.

All in all, the total space and preprocessing time that is required by the explicit levels of the $O(R^2)$ hybrid quad trees is $\tilde{O}(n^2/X^{1+\epsilon} \cdot R^2)$ which is truly subquadratic in $n$ if we set $R=\sqrt{X}$. Moreover, the total query time is $\tilde{O}(nX^{1-\epsilon}/R^2)$ which is truly sublinear in $n$ if we set $R=\sqrt{X}$. Therefore, by setting $R=\sqrt{X}$ we have that the space and preprocessing time of the reduction is truly subquadratic in $n$. Additionally, a query can be answer by at most $O(n/\sqrt{X})$ queries to each SetDisjointness instance plus some additional time that is truly sublinear in $n$.
\end{proof}

\begin{theorem}\label{thm:bounded_setdisjointness_time_clb}
Any solution to SetDisjointness with sets $S_1,S_2,...,S_m \subseteq [u]$ for any value of $u \in [N^{\delta},N]$, such that $N = \sum_{i=1}^{m}{|S_i|}$ and $\delta>0$, must either have $\Omega(m^{2-o(1)})$ preprocessing time or have $\Omega(u^{1/2-o(1)})$ query time, unless the 3SUM Conjecture is false.
\end{theorem}
\begin{proof}
Given an instance of the 3SUM problem that contains 3 arrays $A,B$ and $C$ with $n$ numbers in each of them, we can solve this instance simply by creating a 3SUM indexing instance with arrays $A$ and $B$ and $n$ queries - one for each number in $C$. Thus, using the previous lemma the given 3SUM instance can be reduced for any integer value of $X$ in $[n^{\delta},n]$ (for any $\delta>0$) and for any $\epsilon>0$ to $2\epsilon\log{X}$ instances of SetDisjointness $SD_1,SD_2,...,SD_{2\epsilon\log{X}}$. For any $1 \leq i \leq 2\epsilon\log{X}$, instance $SD_i$ have $N=n\sqrt{u_i}$ elements from universe $[u_i]$ and $m=n\sqrt{\frac{X}{u_i}}$ sets that each one of them is of size $O(\sqrt{u_i})$, where $u_i=X^{1+\epsilon}/2^{i-1}$. The total time for this reduction is $O(n^{2-\epsilon_1})$ for some $\epsilon_1>0$, and the total number of queries is $\tilde{O}(n^2/\sqrt{X})$. Consequently, if we assume to the contradiction that there is an algorithm that solves SetDisjointness on $m$ sets from a universe $[u]$ with $O(m^{2-\epsilon_2})$ preprocessing time for some $\epsilon_2 > 0$ and $O(u^{1/2-\epsilon_3})$ query time for some $0<\epsilon_3 \leq 1/2 $, then we have a solution to 3SUM with $O(n^{2-\epsilon_1}) + \sum_{i=1}^{2\epsilon\log{X}}{O((n\sqrt{\frac{X}{u_i}})^{2-\epsilon_2} + \frac{n^2}{\sqrt{X}}u_i^{1/2-\epsilon_3})}$ time. We have that for any $i$, $u_i \leq X^{1+\epsilon}$ and $\sqrt{\frac{X}{u_i}}\leq \sqrt{\frac{X}{X^{1-\epsilon}}}=X^{\epsilon/2}$. Therefore, $\sum_{i=1}^{2\epsilon\log{X}}{O((n\sqrt{\frac{X}{u_i}})^{2-\epsilon_2} + \frac{n^2}{\sqrt{X}}u_i^{1/2-\epsilon_3}))} =\tilde{O}(n^{(1+\epsilon/2)(2-\epsilon_2)}+\frac{n^2}{\sqrt{X}}X^{(1+\epsilon)(1/2-\epsilon_3)})$. Thus, by setting $\epsilon=\min(\epsilon_2,\epsilon_3)$ we have a total running time that is truly subquadratic in $n$. This contradicts the 3SUM Conjecture.
\end{proof}

Another implication of our reduction in Lemma~\ref{lem:3sum_indexing_setdisjointness_reduction} is a similar reduction from 3SUM to SetIntersection. This reduction leads to a similar conditional lower bound on the preprocessing and query time tradeoff of SetIntersection with bounded universe. This is done in Appendix~\ref{sec:clb_setintersection}.

\section{Applications}\label{sec:applications}

In this section we present several applications of our lower bounds on SetDisjointness and SetIntersection with bounded universe. Several hardness results on the reporting variants of the problems in this section appear in Appendix~\ref{sec:hardness_of_reporting}

\subsection{Range Mode}

As mentioned in the introduction, the range mode problem can be solved using $S$ space and $T$ query time such that: $S \cdot T^2 = \tilde{O}(n^2)$~\cite{CDLMW14,KMS05}. In the following Theorem we prove that $S \cdot T^4 = \tilde{\Omega}(n^2)$. This lower bound is proved based on the Strong SetDisjointness Conjecture using Theorem~\ref{thm:SD_small_universe_LB}. We note that if the lower bound on the query time in Theorem~\ref{thm:SD_small_universe_LB} was $\Omega(u^{1-o(1)})$ instead of $\Omega(u^{1/2-o(1)})$ then the lower bound and upper bound were tight.

\begin{theorem}\label{thm:range_mode_space_clb}
Any data structure that answers Range Mode Queries in $T$ time on a string of length $n$ must use $S=\tilde{\Omega}(n^2/T^4)$ space, unless the Strong SetDisjointness Conjecture is false.
\end{theorem}
\begin{proof}
We use the idea of Chan et al.~\cite{CDLMW14} and apply our theorem on the hardness of SetDisjointness with bounded universe. We begin with an instance of SetIntersection with sets $S_1,S_2,...,S_m \subseteq [u]$ such that $u \in [N^{\delta},N]$, $N = \sum_{i=1}^{m}{|S_i|}$ and $\delta>0$. We create a string $STR$ that is the concatenation of two string $T_1$ and $T_2$ of equal length. The string $T_1$ is the concatenation of the strings $T_{11}, T_{12},...,T_{1m}$. For each $i$ the string $T_{1i}$ is of length $u$ and each character in it is a different number in $[u]$. The prefix of $T_{1i}$ contains all the numbers in $[u] \setminus S_i$ in a sorted order. This prefix is followed by all the numbers in $S_i$ in a sorted order. This is called the suffix of $T_{1i}$. $T_2$ is constructed very similar to $T_1$ but with a change in the order of the suffix and prefix. Specifically, the string $T_2$ is given by the concatenation of the strings $T_{21}, T_{22},...,T_{2m}$. For each $i$ the string $T_{2i}$ is of length $u$ and each character in it is a different number in $[u]$. The prefix of $T_{2i}$ contains all the numbers in $S_i$ in a sorted order. This prefix is followed by all the numbers in $[u] \setminus S_i$ in a sorted order. This is called the suffix of $T_{2i}$. For every $1 \leq i \leq m$, let us denote by $a_i$ the index where the prefix of $T_{1i}$ ends and by $b_i$ the index where the prefix of $T_{2i}$ ends.

The string $STR$ is preprocessed for range mode queries. Then, given a query pair $(i,j)$ for SetDisjointness, we need to decide if $S_i \cap S_j=\emptyset $ or not. This is done by a range mode query for the range $[a_i + 1, b_j]$. For every $p \in [2]$ and $q \in [m]$, the string $T_{pq}$ contains characters that represent all the numbers in $[u]$, such that each of these numbers occurs exactly once in the string. Between $T_{1i}$ and $T_{2j}$ we have $m-i+j-1$ substrings that each of them contains all the characters from $[u]$. Therefore, each character occurs $m-i+j-1$ times between $T_{1i}$ and $T_{2j}$. The suffix of $T_{1i}$ starting at index $a_i+1$ contains all the characters that represent the elements of $S_i$, while the prefix of $T_{2j}$ ending at index $b_j$ contains all the characters that represent the elements of $S_j$. Consequently, if there is an intersection between $S_i$ and $S_j$ we will have at least one character that occurs in both the suffix of $T_{1i}$ and the prefix of $T_{2j}$. Thus, the mode of the range $[a_i + 1, b_j]$ will be $m-i+j+1$ if $S_i \cap S_j \neq \emptyset$, and less than $m-i+j+1$ if the $S_i \cap S_j = \emptyset$. Therefore, if we get from the range mode query a character $c$ that occurs $m-i+j+1$ times in the query range we know that the intersection is not empty, and if not we know that the intersection is empty. Even if the range mode query does not return the frequency of the mode within the query range, but rather just the mode element itself, we can save a hash table for every input set and use this tables to check in constant time if the returned element occurs in both $S_i$ and $S_j$.

Consequently, an instance of SetDisjointness with $m$ sets from universe $[u]$ (such that $u \in [N^{\delta},N]$, $N = \sum_{i=1}^{m}{|S_i|}$ and $\delta>0$), can be reduced to an instance of the range mode problem with a string of length $n=2mu$, such that every query to the SetDisjointness instance can be answered by a query to the range mode instance. Let us assume to the contrary that the range mode problem can be solved by a data structure that answers queries in $\tilde{O}(T)$ time per query using $\tilde{O}(S)$ space such that $S \cdot T^4 = \tilde{O}(n^{2-\epsilon})$. Let $T=\tilde{O}(u^{1/2-\epsilon/4})$, we have that $S=\tilde{O}(n^{2\epsilon}/T^4) = \tilde{O}((mu)^{2-\epsilon}/u^{4(1/2-\epsilon/4)}) = \tilde{O}(m^{2-\epsilon}u^{2-\epsilon}/u^{2-\epsilon)}) = \tilde{O}(m^{2-\epsilon})$. Therefore, we have a solution to SetDisjointness with $m$ sets from universe $[u]$ with query time $\tilde{O}(u^{1/2-\epsilon/4})$ and space $\tilde{O}(mu + m^{2-\epsilon})$ (we add $mu$ to the space usage, as we must at least save the string $ST$). According to Corollary~\ref{thm:SD_small_universe_LB_for_specific_m} the reduction from general SetDisjointness to SetDisjointness with bounded universe holds for $N/\sqrt{u} \leq m$. Therefore, for any value of $u \leq N^{2/3-\epsilon}$ we have that $\sqrt{u} \leq N^{1/3-\epsilon/2}$. Thus, the following holds: $\sqrt{u} \leq N^{1/3-\epsilon/2} \Rightarrow \frac{1}{N^{1/3-\epsilon/2}} \leq \frac{1}{\sqrt{u}} \Rightarrow \frac{N}{N^{1/3-\epsilon/2}} \leq \frac{N}{\sqrt{u}} \Rightarrow N^{2/3+\epsilon/2} \leq \frac{N}{\sqrt{u}} \leq m \Rightarrow N^{2/3+\epsilon/2-\frac{2}{3}\epsilon-\frac{\epsilon^2}{2}} \leq m^{1-\epsilon}$. Consequently, we have that $u \leq N^{2/3-\epsilon} < N^{2/3-\epsilon/6-\epsilon/2} \leq  N^{2/3-\epsilon/6-\epsilon^2/2} \leq m$. All in all, for any $u\leq N^{2/3-\epsilon}$ the reduction holds and $mu=\tilde{O}(m^{2-\epsilon})$. Consequently, the total space for solving SetDisjointness with bounded universe using our reduction to the range mode problem is $\tilde{O}(m^{2-\epsilon})$ and the query time is $\tilde{O}(u^{1/2-\epsilon/4})$. This contradicts the Strong SetDisjointness Conjecture according to Corollary~\ref{thm:SD_small_universe_LB_for_specific_m}.
\end{proof}

Using Theorem~\ref{thm:bounded_setdisjointness_time_clb} and the same idea from the proof of Theorem~\ref{thm:range_mode_space_clb}, we obtain the following result regarding the preprocessing and query time tradeoff for solving the range mode problem:
\begin{corollary}\label{crl:range_mode_time_clb}
Any data structure that answers Range Mode Queries in $T$ time on a string of length $n$ must have $P=\tilde{\Omega}(n^2/T^4)$ preprocessing time, unless the 3SUM Conjecture is false.
\end{corollary}

\subsection{Distance Oracles}
Agarwal~\cite{Agarwal14} presented space-time tradeoffs for distance oracles for undirected graph $G=(V,E)$ with average degree $\mu$ (that is, $\mu=\frac{2|E|}{|V|}$): (i) $(1+\frac{1}{k})$-stretch distance oracles that use $\tilde{O}(|E|+\frac{|V|^2}{\alpha})$ space and have $O((\alpha\mu)^k)$ query time, for any $1 \leq \alpha \leq |V|$ (ii) (1+$\frac{1}{k+0.5})$-stretch distance oracles that use $\tilde{O}(|E|+\frac{|V|^2}{\alpha})$ space and have $O(\alpha(\alpha\mu)^k)$ query time, for any $1 \leq \alpha \leq |V|$. (iii) $(1+\frac{2}{3})$-stretch distance oracle that uses $\tilde{O}(|E|+\frac{|V|^2}{\alpha})$ space and has $O(\alpha\mu)$ query time for any $1 \leq \alpha \leq (\frac{|V|^2}{|E|})^{\frac{1}{3}}$. In the last result ((iii)) Agarwal managed to shave an $\alpha$ factor of the query time in (ii) (for $k=1$). Therefore, both $\frac{5}{3}$-stretch distance oracle and $2$-stretch distance oracle (by setting $k=1$ in (i)) have the same space-time tradeoff. It is known that $3$-stretch distance oracle has a better tradeoff (see~\cite{AGH11}). Moreover, by (i) and (ii) the tradeoff for stretch less than $5/3$ gets worse as the stretch guarantee is better. Thus, it seems natural to expect a better tradeoff for stretch more than $5/3$ and less-than-equal to $2$.

In the following theorem we prove that improving the tradeoff of Agarwal~\cite{Agarwal14} is impossible for any stretch $t\in [\frac{2}{3},2)$, unless the Strong SetDisjointness Conjecture is false:

\begin{theorem}\label{distance_oracles_space_clb}
Any distance oracle for undirected graph $G=(V,E)$ with stretch less than $2$ must either use $\Omega(|V|^{2-o(1)})$ space or have $\Omega(\mu^{1-o(1)})$ query time, where $\mu$ is the average degree of a vertex in $G$, unless Strong SetDisjointness Conjecture is false.
\end{theorem}
\begin{proof}
We use the idea of Cohen and Porat~\cite{CP102} with our hardness results for SetDisjointness with bounded universe. Given an instance of SetDisjointness with sets $S_1,S_2,...,S_m \subseteq [u]$ such that $u \in [N^{\delta},N]$, $N = \sum_{i=1}^{m}{|S_i|}$ and $\delta>0$, we construct a bipartite graph $G=(V,E)$ as follows: In one side, we create a vertex $v_i$ for each set $S_i$. In the other side, we create a vertex $u_j$ for each element $j \in [u]$. For each element $x$ in some set $S_i$ we create an edge $(v_i,u_x)$. Formally, $V=\{v_i| 1 \leq i \leq m\} \cup \{u_j| j \in [u]\}$ and $E=\{(v_i,u_x) | x \in S_i\}$. For any $i,j \in [m]$, if $S_i \cap S_j \neq \emptyset$ then it is clear that the distance between $v_i$ and $v_j$ is exactly $2$. Otherwise, the distance is at least $4$. A stretch less-than $2$ distance oracle can distinguish between these two possibilities and therefore a SetDisjointness query can be answered by one query to a stretch less-than $2$ distance oracle for $G$.

It is clear that $|V| = m + u$ and $|E|=N$. We assume to the contradiction that there is a stretch less than two distance oracle that uses $\tilde{O}(|V|^{2-\epsilon_1})$ space and answers queries in $\tilde{O}(\mu^{1-\epsilon_2}) = \tilde{O}((\frac{|E|}{|V|})^{1-\epsilon_2})$ time, for some $\epsilon_1,\epsilon_2>0$. Therefore, SetDisjointness with bounded universe can be solved using $\tilde{O}((m+u)^{2-\epsilon_1})$ space and queries can be answered using $\tilde{O}((\frac{N}{m+u})^{1-\epsilon_2})$ time. According to Corollary~\ref{thm:SD_small_universe_LB_for_specific_m} the reduction from general SetDisjointness to SetDisjointness with bounded universe holds for $N/\sqrt{u} \leq m$. Therefore, for any value of $u \leq N^{2/3}$ we have that the reduction holds and $u \leq m$ (see the full details in the proof of Theorem~\ref{thm:range_mode_space_clb}). Moreover, we have that $N/(m+u) \leq N/m \leq \sqrt{u}$. Consequently, for any $u \leq N^{2/3}$ we have a solution to SetDisjointness with bounded universe that uses $\tilde{O}((m+u)^{2-\epsilon_1}) = \tilde{O}(m^{2-\epsilon_1})$ space and answers queries in $\tilde{O}(\frac{N}{m+u}^{1-\epsilon_2}) = \tilde{O}((\sqrt{u})^{1-\epsilon_2}) = \tilde{O}(u^{1/2-\epsilon_2/2})$ time. This contradicts Strong SetDisjointness Conjecture according to Corollary~\ref{thm:SD_small_universe_LB_for_specific_m}.
\end{proof}

The previous theorem can be stated in a different way that makes it clear that the space-time tradeoff of Agarwal~\cite{Agarwal14} is tight for distance oracles with stretch $t$ such that $5/3 \leq t < 2$.

\begin{corollary}
There is no stretch less-than-$2$ distance oracle for undirected graph $G=(V,E)$ that uses $\tilde{O}(\frac{|V|^2}{\alpha})$ space and have $\tilde{O}(\alpha^{1-\epsilon}\mu)$ query time for any $|V|^{\delta} \leq \alpha$ and any $\delta,\epsilon>0$, unless conjecture 1 is false.
\end{corollary}

Using Theorem~\ref{thm:bounded_setdisjointness_time_clb} and the same idea from the proof of Theorem~\ref{distance_oracles_space_clb}, we obtain the following result regarding the preprocessing and query time tradeoff for distance oracles with stretch less-than-2:

\begin{theorem}
Any distance oracle for undirected graph $G=(V,E)$ with stretch less than $2$ must either be constructed in $\Omega(|V|^{2-o(1)})$ preprocessing time or have $\Omega(\mu^{1-o(1)})$ query time, where $\mu$ is the average degree of a vertex in $G$, unless the 3SUM Conjecture is false.
\end{theorem}

\subsection{3SUM-Indexing with Small Universe}

In the following theorem we prove a conditional lower bound on the space-time tradeoff for solving 3SUM-Indexing with universe size that is $[n^{2+\epsilon}]$ for any $\epsilon>0$ ($n$ is the size of the input arrays).

\begin{theorem}\label{3sum_indexing_space_clb}
For any $\epsilon>0$ and $0 < \delta \leq 1$, any solution to 3SUM-Indexing with arrays $A=a_1,a_2,...,a_n$ and $B=b_1,b_2,...,b_n$ such that for every $i \in [n]$ $a_i,b_i \in [n^{2+\epsilon}]$ must either use $\Omega(n^{2-\delta-o(1)})$ space or have $\Omega(n^{\frac{\delta}{2}-o(1)})$ query time, unless Strong SetDisjointness Conjecture is false.
\end{theorem}
\begin{proof}
We use the idea of Goldstein et al.~\cite{GKLP17} with our hardness for SetDisjointness with bounded universe. We begin with an instance of SetDisjointness with sets $S_1,S_2,...,S_m \subseteq [u]$ such that $u=N^{\delta}$, $m \in [\frac{N}{u^{1/2}},\frac{N}{u^{1/2-\epsilon'}}]$, $N = \sum_{i=1}^{m}{|S_i|}$, $\epsilon'=\epsilon/2$ and $\delta>0$.

For every element $x$ in some set $S_i$ we create two numbers $x_{1,i}$ and $x_{2,i}$. The number $x_{1,i}$ consists of 3 blocks of bits (ordered from the least significant bit toward the most significant bit): (i) A block of $\log{m}$ bits that contains the value of the index $i$. (ii) A block of $\log{m}$ padding zero bits. (iii) A block of $\log{u}$ bits that contains the value of $x-1$. The number $x_{2,i}$ consists of 3 blocks of bits (ordered from the least significant bit toward the most significant bit): (i) A block of $\log{m}$ padding zero bits. (ii) A block of $\log{m}$ bits that contains the value of the index $i$. (iii) A block of $\log{u}$ bits that contains the value of $u-x$. We place the number $x_{1,i}$ in array $A$ and the number $x_{2,i}$ in array $B$. The number of elements in each of these arrays is $N$, as we add a number to each array for every element in the input sets. These two arrays form an instance of 3SUM-Indexing which is preprocessed in order to answer queries.

Given a query asking whether $S_i \cap S_j = \emptyset$ or not, we can answer it by creating a query number $z$ to the 3SUM-Indexing instance as follows: The number $z$ consists of 3 blocks of bits (ordered from the least significant bit toward the most significant bit): (i) A block of $\log{m}$ bits that contain the value of the index $i$. (ii) A block of $\log{m}$ that contain the value of the index $j$. (iii) A block of $\log{u}$ bits that contains the value of $u-1$. It straightforward to see that we get a positive answer to the query number $z$ iff $S_i \cap S_j \neq \emptyset$: (i) If we have $x_{1,k_1} \in A$ and $y_{2,k_2} \in B$ such that $x_{1,k_1}+y_{2,k_2}=z$, then we must have that: (1) $k_1=i$ which means that $x$ is in $S_i$. (2) $k_2=j$ which means that $y$ is in $S_j$. (3) $x-1+u-y=u-1$ which means that $x=y$. (ii) If $S_i \cap S_j \neq \emptyset$ then there is an element $x$ such that $x \in S_i$ and $x \in S_j$. From our construction it is clear that indeed $x_{1,i}+x_{2,j}=z$.

Thus, we have reduced our SetDisjointness instance to an instance of 3SUM-Indexing such that each query to the SetDisjointness instance can be answered by a query to the 3SUM-Indexing instance. The size of each array in the 3SUM-Indexing instance is $N$. All the numbers in these arrays have $2\log{m}+\log{u}$ bits. Let $u=N^{\delta}$ and $m \in [\frac{N}{u^{1/2}},\frac{N}{u^{1/2-\epsilon'}}]$, for $\epsilon' \leq \frac{\epsilon}{2}$, then the number of bits in each number of $A$ and $B$ is bounded by $2\log{\frac{N}{u^{1/2-\epsilon'}}}+\log{N^{\delta}}=2\log{N^{1-\delta/2+\epsilon'\delta}}+\log{N^{\delta}}= 2(1-\delta/2+\epsilon'\delta)\log{N}+\delta\log{N}=(2+2\epsilon'\delta)\log{N} \leq (2+\epsilon)\log{N}$. By setting $n=N$ we have that both $A$ and $B$ have $n$ elements and all the numbers are in $[n^{2+\epsilon}]$.

We assume to the contradiction that 3SUM-Indexing with universe $[n^{2+\epsilon}]$ can be solved using $\tilde{O}(n^{2-\delta-\gamma_1})$ space, while answering queries in $\tilde{O}(n^{\frac{\delta}{2}-\gamma_2})$ time, for some $\gamma_1,\gamma_2>0$. Following our reduction this means that we can solve SetDisjointness with $m$ from universe $[u]$ using $S=\tilde{O}(n^{2-\delta-\gamma_1})$ space, while answering queries in $T=\tilde{O}(n^{\frac{\delta}{2}-\gamma_2})$ time. We have that $u=n^{\delta}$, so $n=u^{1/\delta}$. Moreover, $m \geq n^{1-\delta/2}$, so $n \leq m^{1/(1-\delta/2)}$. Therefore, $S=\tilde{O}(m^{(2-\delta-\gamma_1)/(1-\delta/2)})=\tilde{O}(m^{2-\frac{\gamma_1}{1-\frac{\delta}{2}}})$ and $T=\tilde{O}(u^{(\frac{\delta}{2}-\gamma_2)/\delta})=\tilde{O}(u^{1/2-\gamma_2/\delta})$. This contradicts Corollary~\ref{thm:SD_small_universe_LB_for_specific_m}.
\end{proof}

Using Theorem~\ref{thm:bounded_setdisjointness_time_clb} and the same idea from the proof of Theorem~\ref{3sum_indexing_space_clb}, we obtain the following result regarding the preprocessing and query time tradeoff for distance oracles with stretch less-than-2:

\begin{theorem}
For any $\epsilon>0$ and $0 < \delta \leq 1$, any solution to 3SUM-Indexing with arrays $A=a_1,a_2,...,a_n$ and $B=b_1,b_2,...,b_n$ such that for every $i \in [n]$ $a_i,b_i \in [n^{2+\epsilon}]$ must either have $\Omega(n^{2-\delta-o(1)})$ preprocessing time or have $\Omega(n^{\frac{\delta}{2}-o(1)})$ query time, unless the 3SUM Conjecture is false.
\end{theorem}

\bibliographystyle{plain} \bibliography{ms}

\newpage
\appendix
\section*{Appendix}

\section{Conditional Lower Bounds for SetIntersection}\label{sec:clb_setintersection}

In the following theorem we prove a conditional lower bound on SetIntersection with bounded universe based on the Strong SetDisjointness Conjecture by generalizing the ideas from Theorem~\ref{thm:SD_small_universe_LB}. Specifically, we demonstrate that for SetIntersection we either have the same space lower bound as for SetDisjointness or we have a $\tilde{\Omega}(u^{1-o(1)} + out)$ bound on the query time. The query time bound is stronger than the $\Omega(u^{1/2-o(1)})$ bound that we have for SetDisjointness. However, we argue that this lower bound for SetIntersection holds only when the output is large. If we have an upper bound on the size of the output we still have a lower bound on the query time, but this lower bound gets closer to $\tilde{\Omega}(u^{1/2-o(1)} + out)$ as the size of the output gets smaller. Eventually, this coincides with the lower bound we have for SetDisjointness (notice that in order to answer SetDisjointness queries we just need to output a single element from the intersection if there is any).

\setcounter{theorem}{14}

\begin{theorem}\label{thm:SI_small_universe_LB_from_SD}
Any solution to SetIntersection with sets $S_1,S_2,...,S_m \subseteq [u]$ for any value of $u \in [N^{\delta},N]$, such that $N = \sum_{i=1}^{m}{|S_i|}$ and $\delta>0$, must either use $\Omega(m^{2-o(1)})$ space or have $\tilde{\Omega}(u^{\alpha-o(1)} + out)$ query time, for any $1/2 \leq \alpha \leq 1$ and any output size $out$ such that $out=\Omega(u^{2\alpha-1-\delta})$ and $\delta>0$ ,unless Strong SetDisjointness Conjecture is false.
\end{theorem}

\begin{proof}
We use the same idea as in the proof of Theorem~\ref{thm:SD_small_universe_LB}. Let us assume to the contradiction that Strong SetDisjointness Conjecture is true but there is an algorithm $A'$ that solves SetIntersection on $m$ sets from a universe $[u]$ and creates a data structure $D$, such that the space complexity of the data structure $D$ is $O(m^{2-\epsilon_1})$ for some $\epsilon_1 > 0$ and the query time of algorithm $A'$ is $O(u^{\alpha-\epsilon_2})$ for some $0<\epsilon_2 \leq 1/2 $. We define $\epsilon=\min(\epsilon_1,\epsilon_2)$.

In the proof, we call those sets with at least $u^{\alpha-3/4\epsilon}$ elements \emph{large sets} and those sets with at most $O(u^{\alpha-\epsilon})$ elements \emph{small sets}. All other sets are called \emph{medium sets}.

SetIntersection (for general universe) can be solved in the following way:

The preprocessing phase is similar to the one that is done in the proof of Theorem~\ref{thm:SD_small_universe_LB} with the following changes: 1. In step (5) we check for each pair of medium sets $S_i$ and $S_j$ such that $S_i \cap S_j = \emptyset$ that the size of $h_k(S_i) \cap h_k(S_j)$ is no more than $u^{2\alpha-1-3/2\epsilon}$ for at least one $h_k:U \rightarrow [8u]$ that we pick in step (4). This is done instead of just checking for the emptiness of $h_k(S_i) \cap h_k(S_j)$. 2. In step (6.2) we use algorithm $A'$ to create a data structure $D_k$ that solves the SetIntersection problem instead of the SetDisjointness problem.

The query phase is also very similar to the one from Theorem~\ref{thm:SD_small_universe_LB} with the following change: In step (3), for each $k$, we get one by one the elements in the intersection of $h_k(S_i)$ and $h_k(S_j)$ by querying the data structure $D_k$. For each element $e$ in that intersection we verify that it is contained in both $S_i$ and $S_j$ using the tables $T_i$ and $T_j$. If this is the case, then we return that the sets are not disjoint. Otherwise, we add one to a counter of the number of elements in $(h_k(S_i) \cap h_k(S_j))\setminus(S_i \cap S_j)$. If this counter exceeds $u^{2\alpha-1-3/2\epsilon}$ we stop the query immediately and continue to the next value of $k$.

The correctness of this reduction follows from the same arguments as in the proof of Theorem~\ref{thm:SD_small_universe_LB}. The difference is in analysing the hash functions and their properties. For any two unequal elements $x_1\in S_i$ and $x_2 \in S_j$, where both $S_i$ and $S_j$ are medium sets, and for any hash function $h_k:N \rightarrow [8u]$ we have that $\Pr[h_k(x_1)=h_k(x_2)] \leq 1/(8u)$. We call two unequal elements $x_1\in S_i$ and $x_2 \in S_j$ such that $h_k(x_1)=h_k(x_2)$ a false-positive of $h_k$. The number of elements in the medium sets is no more than $u^{\alpha-3/4\epsilon}$. Consequently, the expected number of false-positives in $h_k(S_i) \cap h_k(S_j)$ is no more than $(u^{\alpha-3/4\epsilon})^2/8u = u^{2\alpha-1-3/2\epsilon}/8$. By Markov inequality the probability that the number of false-positives for a specific $h_k$ is more than $u^{2\alpha-1-3/2\epsilon}$ is no more than 1/8. Therefore, the probability that a pair of medium sets $S_i$ and $S_j$ has more than $u^{2\alpha-1-3/2\epsilon}$ false-positives when applying $h_k$ for all $k \in [\log{n}]$ is no more than $(1/8)^{\log{n}} = 1/n^3$. Thus, the probability that the number of false-positives for any pair of medium sets is more than $u^{2\alpha-1-3/2\epsilon}$ when applying $h_k$ for all $k \in [\log{n}]$ is no more than $n^2/n^3 = 1/n$ by the union-bound. Using the probability method we get that there must be $\log{n}$ hash functions such that for every pair of medium sets $S_i$ and $S_j$ the number of false-positives is no more than $u^{2\alpha-1-3/2\epsilon}$ after applying the hash functions by at least one of the $\log{n}$ hash functions.

\textbf{Complexity analysis}.

\textbf{Space complexity}. The space for the tables in step (1) of the preprocessing is clearly $O(N)$ - linear in the total number of elements. The total number of large sets $d$ is at most $O(N/u^{\alpha-3/4\epsilon})$. The total number of medium sets $e$ is at most $O(N/u^{\alpha-\epsilon})$. Therefore, the size of the matrix $M$ is at most $O(N/u^{\alpha-3/4\epsilon} \cdot N/u^{\alpha-\epsilon}) = O(N^2/u^{2\alpha-7/4\epsilon})$. There are $\log{n}$ data structures that are created in step (6). Each data structure uses at most $O((N/u^{\alpha-\epsilon})^{2-\epsilon}) = O(N^{2-\epsilon}/u^{2\alpha-(2+\alpha)\epsilon+\epsilon^2})$ space. Consequently, the total space complexity is $S=\tilde{O}(N^2/u^{2\alpha-7/4\epsilon}+N^{2-\epsilon}/u^{2\alpha-(2+\alpha)\epsilon+\epsilon^2})$.

\textbf{Query time complexity}. Step (1) of the query algorithm can be done in $O(u^{\alpha-\epsilon})$ as this is the size of the largest small set. Step (2) is done in constant time by looking at the right position in $M$. In step (3) we do $\log{n}$ queries using algorithm $A'$ and the data structures $D_k$. the universe of the sets after applying any hash function $h_k$ is $[8u]$, so the query time for each query is $O(u^{\alpha-\epsilon}+out)$ ($out$ is the size of the output we get from the query). We do not allow the query to output more than $u^{2\alpha-1-3/2\epsilon} < u^{\alpha-\epsilon}$ elements. Therefore, the total query time is $T=O(u^{\alpha-\epsilon})$.

\smallskip

Following our analysis we have that $S \cdot T^2 = \tilde{O}((N^2/u^{2\alpha-7/4\epsilon}+N^{2-\epsilon}/u^{2\alpha-(2+\alpha)\epsilon+\epsilon^2}) \cdot (u^{\alpha-\epsilon})^2) = \tilde{O}(N^2u^{-1/4\epsilon}+N^{2-\epsilon}u^{\alpha\epsilon-\epsilon^2})$. As $\alpha \leq 1$ and $u \leq N$, we have that  $u^{\alpha\epsilon} \leq N^{\epsilon}$. Therefore, $S \cdot T^2 = \tilde{O}(N^2u^{-1/4\epsilon}+N^2u^{-\epsilon^2})$.
This contradicts the Strong SetDisjointness Conjecture and therefore our assumption is false.
\end{proof}

A better lower bound on the space complexity for solving SetIntersection can be obtained based on the Strong SetIntersection Conjecture. This is demonstrated by the following theorem:

\begin{theorem}
Any solution to SetIntersection with sets $S_1,S_2,...,S_m \subseteq [u]$ for any value of $u \in [N^{\delta},N]$, such that $N = \sum_{i=1}^{m}{|S_i|}$ and $\delta>0$, must either use $\Omega((m^2u^{\alpha})^{1-o(1)})$ space or have $\tilde{\Omega}(u^{\alpha-o(1)} + out)$ query time for any $1/2 \leq \alpha \leq 1$ and any output size $out$ such that $out=\Omega(u^{2\alpha-1-\delta})$ and $\delta>0$, unless Strong SetIntersection Conjecture is false.
\end{theorem}
\begin{proof}
The proof is very similar to the proof of Theorem~\ref{thm:SD_small_universe_LB}. Let us assume to the contradiction that the Strong SetIntersection Conjecture is true but there is an algorithm $A'$ that solves SetIntersection on $m$ sets from a universe $[u]$ and creates a data structure $D$, such that the space complexity of the data structure $D$ is $O(m^2u^{\alpha})^{1-\epsilon_1})$ for some $\epsilon_1 > 0$ and the query time of algorithm $A'$ is $O(u^{\alpha-\epsilon_2})$ for some $0<\epsilon_2 \leq 1/2 $. We define $\epsilon=\min(\epsilon_1,\epsilon_2)$.

In order to solve SetIntersection for general universe we use almost the same preprocessing and query procedures as in the the proof of Theorem~\ref{thm:SD_small_universe_LB} except for the following changes: 1. In the preprocessing phase, we do not save in matrix $M$ in the entries $M[p(i),p(j)]$ or $M[p(i),d+q(j)]$ just the answer to the emptiness of the intersection of $S_i$ and $S_j$, but rather we save in this location a list of all the elements within the intersection of $S_i$ and $S_j$. 2. In the query phase, in step (2) we return a list of elements and not just a single bit. 3. In the query phase, in step (3) for each $k$ we get the intersection of $h_k(S_i)$ and $h_k(S_j)$ by querying the data structure $D_k$. For each element $e$ in that intersection we return it after verifying that it is contained in both $S_i$ and $S_j$ using the tables $T_i$ and $T_j$. Moreover, we count the number of elements in $(h_k(S_i) \cap h_k(S_j))\setminus(S_i \cap S_j)$ as we get them from the query and if they exceed $u^{2\alpha-1-3/2\epsilon}$ we stop the query immediately and continue with the next value of $k$.

The correctness of the above solution to set intersection follows from the same arguments as in the proof of Theorem~\ref{thm:SI_small_universe_LB_from_SD}.

\textbf{Complexity analysis}.

\textbf{Space complexity}. The space for the tables in step (1) of the preprocessing is clearly $O(N)$. Matrix $M$ in this solution contains in each entry the complete list of elements in the intersection of some pair of sets. The total number of large sets $d$ is at most $O(N/u^{\alpha-3/4\epsilon})$. The total number of medium sets $e$ is at most $O(N/u^{\alpha-\epsilon})$. The total number of elements in all sets is $N$. Therefore, the size of the matrix $M$ is at most $O(N/u^{\alpha-3/4\epsilon} \cdot N/u^{\alpha-\epsilon} \cdot u^{\alpha-\epsilon}) = O(N^2/u^{\alpha-3/4\epsilon})$ (see the full details in Appendix~\ref{sec:algorithms_for_setintersection}). There are $\log{n}$ data structures that are created in step (6). Each data structure use at most $O(((N/u^{\alpha-\epsilon})^2u^{\alpha})^{1-\epsilon}) = O(N^{2-2\epsilon}/u^{\alpha-(2+\alpha)\epsilon+2\epsilon^2})$ space. Consequently, the total space complexity is $S=\tilde{O}(N^2/u^{\alpha-3/4\epsilon}+N^{2-2\epsilon}/u^{\alpha-(2+\alpha)\epsilon+2\epsilon^2})$.

\textbf{Query time complexity}. Step (1) of the query algorithm can be done in $O(u^{\alpha-\epsilon})$ as this is the size of the largest small set. Step (2) is done in constant time plus the output size by looking at the right position in $M$. In step (3) we do $\log{n}$ queries using algorithm $A'$ and the data structures $D_k$. The universe of the sets after applying any hash function $h_k$ is $[u]$, so the query time for each query is $O(u^{\alpha-\epsilon}+out)$ ($out$ is the size of the output we get from the query). We do not allow the query to output more than $u^{2\alpha-1-3/2\epsilon} < u^{\alpha-\epsilon}$ false-positive elements. Therefore, the total query time is $O(T+out)$, where $T=O(u^{\alpha-\epsilon})$.

\smallskip

Following our analysis we have that $S \cdot T = \tilde{O}((N^2/u^{\alpha-3/4\epsilon}+N^{2-2\epsilon}/u^{\alpha-(2+\alpha)\epsilon+2\epsilon^2}) \cdot (u^{\alpha-\epsilon})) = \tilde{O}(N^2u^{-1/4\epsilon}+N^{2-2\epsilon}u^{(1+\alpha)\epsilon-2\epsilon^2})$. As $\alpha \leq 1$ and $u \leq N$, we have that  $u^{(1+\alpha)\epsilon} \leq N^{2\epsilon}$. Therefore, $S \cdot T = \tilde{O}(N^2u^{-1/4\epsilon}+N^2u^{-2\epsilon^2})$.
This contradicts the Strong SetIntersection Conjecture and therefore our assumption is false.
\end{proof}

The construction in the proof of Lemma~\ref{lem:3sum_indexing_setdisjointness_reduction} can be modified in order to obtain the following reduction from 3SUM-Indexing to SetIntersection:

\setcounter{theorem}{16}

\begin{lemma}\label{lem:3sum_indexing_setintersection_reduction}
For any $0 < \gamma <\delta \leq 1$, an instance of 3SUM-Indexing that contains 2 arrays with $n$ integers can be reduced to an instance $SI$ of SetIntersection. The instance $SI$ have $N=n\sqrt{u}$ elements from universe $[u]$ and $m=n^{1+\gamma-\delta/2}$ sets that each one of them is of size $O(\sqrt{u})$, where $u=n^{\delta}$ and $0 < 2\gamma <\delta \leq 1$. The time and space complexity of the reduction is $\tilde{O}(n^{2+2\gamma-\delta})$. Each query to the 3SUM-Indexing instance can be answered by at most $O(n^{1+\gamma-\delta})$ queries to $SI$ plus some additional $O(\log{n})$ time.
\end{lemma}

\begin{proof}
We follow the construction from the proof of Lemma~\ref{lem:3sum_indexing_setdisjointness_reduction}. In each quad tree we construct for some two buckets $A_i$ and $B_j$, we save the convolution results of the corresponding sub-vectors until the bottom level in which the size of each subvector is $X$. In this level, for each pair of sub-vectors we create $O(\sqrt{X})$ sets (representing different shifts) in the same way we construct the sets for the SetDisjointness instances in the proof of Lemma~\ref{lem:3sum_indexing_setdisjointness_reduction}. These sets form a SetIntersection instance that contains $O(R \cdot n/X \cdot \sqrt{X})=O(nR/\sqrt{X})$ sets. In the query phase, whenever we search a quad tree and get to a leaf node we can immediately report all pairs of elements that are witnesses for $C_{i,j}[h_2(z)]$. This is easily done by a single SetIntersection query. The number of sub-vectors in the bottom level is $O(n/X)$ for both $v_{A_i}$ and $v_{B_j}$. For every sub-vector of $v_{A_i}$ there are at most $O(1)$ sub-vectors of $v_{B_j}$ that their convolution with $v_{A_i}$ may contain a witness pair for $C_{i,j}[h_2(z)]$. Consequently, we do at most $O(n/X)$ intersection queries within each quad tree.

Therefore, the total space for constructing the quad trees' levels with explicit convolution results is $\tilde{O}(n^2/X \cdot R^2)$ (see the full analysis in the proof of Lemma~\ref{lem:3sum_indexing_setdisjointness_reduction}). This is also the preprocessing time for constructing these quad trees as the convolution of two $n$-length vectors can be calculated in $\tilde{O}(n)$ time. It is clear that the space and preprocessing time are truly subquadratic in $n$ for any $\delta>2\gamma>0$. Moreover, the query time overhead is no more than $O(\log{n})$ for every query (a search through a path from the root to a leaf in some quad tree).
\end{proof}

\begin{theorem}\label{thm:bounded_setintersection_time_clb}
Any solution to SetIntersection with sets $S_1,S_2,...,S_m \subseteq [u]$ for any value of $u \in [N^{\delta},N]$, such that $N = \sum_{i=1}^{m}{|S_i|}$ and $\delta>0$, must either have $\Omega(m^{2-o(1)})$ preprocessing time or have $\tilde{\Omega}(u^{1-o(1)}+out)$ query time, unless the 3SUM Conjecture is false.
\end{theorem}
\begin{proof}
Given an instance of the 3SUM problem that contains 3 arrays $A,B$ and $C$ with $n$ numbers in each of them, we can solve this instance simply by creating a 3SUM indexing instance with arrays $A$ and $B$ and $n$ queries - one for each number in $C$. Thus, using the previous lemma the given 3SUM instance can be reduced to an instance of SetIntersection with $m=n^{1+\gamma-\delta/2}$ sets from universe $[u]$ using $O(n^{2+2\gamma-\delta})$ time for preprocessing, where the total number of queries to these instances is $O(n^{2+\gamma-\delta})$.

We assume to the contradiction that there is an algorithm that solves SetIntersection on $m$ sets from a universe $[u]$ with $O(m^{2-\epsilon_1})$ preprocessing time for some $\epsilon_1 > 0$ and $O(u^{1-\epsilon_2}+out)$ query time for some $0<\epsilon_2 \leq 1 $. If we choose the value of $\delta$ such that $\delta>\max{2,\frac{1}{\epsilon_2}}$, then we have a solution to 3SUM with truly subquadratic running time. This contradicts the 3SUM Conjecture.
\end{proof}

\newpage

\section{Hardness of Reporting Problems}\label{sec:hardness_of_reporting}
\subsection{Range Mode Reporting}
In the reporting variant of the Range Mode problem we are required to report all elements in the query range that are the mode of this range. We have stronger lower bounds for this variant using the same construction as in the proof of Theorem~\ref{thm:range_mode_space_clb} with the conditional lower bounds for SetIntersection with bounded universe. The results refer to both the interplay between space and query time and the interplay between preprocessing and query time.
\begin{theorem}\label{thm:range_mod_reportinge_space_clb}
Any data structure that answers Range Mode Reporting in $O(T+out)$ time on a string of length $n$, where $out$ is the output size, must use $S=\tilde{\Omega}(n^2/T^2)$ space, unless the Strong SetIntersection Conjecture is false.
\end{theorem}

\begin{theorem}\label{thm:range_mode_reporting_time_clb}
Any data structure that answers Range Mode Reporting in $O(T+out)$ time on a string of length $n$, where $out$ is the output size, must must have $P=\tilde{\Omega}(n^2/T^2)$ preprocessing time, unless the 3SUM Conjecture is false.
\end{theorem}

\subsection{3SUM-Indexing Reporting}
In the reporting variant of 3SUM-Indexing we are required to report all pairs of numbers $a\in A$ and $b\in B$ such that their sum equals the query number. Using our hardness results for SetIntersection with bounded universe we prove the following conditional lower bounds on 3SUM-Indexing reporting. These results are obtained by applying the same techniques as in the proof of Theorem~\ref{3sum_indexing_space_clb}.

\begin{theorem}\label{3sum_indexing_reporting_space_clb}
For any $\epsilon>0$ and $0 < \delta \leq 1$, any solution to 3SUM-Indexing reporting with arrays $A=a_1,a_2,...,a_n$ and $B=b_1,b_2,...,b_n$ such that for every $i \in [n]$ $a_i,b_i \in [n^{2+\epsilon-\delta}]$ must either use $\Omega(n^{2-\delta-o(1)})$ space or have $\tilde{\Omega}(n^{\delta-o(1)} + out)$ query time, where $out$ is the output size, unless Strong SetIntersection Conjecture is false.
\end{theorem}

\begin{theorem}\label{3sum_indexing_reporting_time_clb}
For any $\epsilon>0$ and $0 < \delta \leq 1$, any solution to 3SUM-Indexing reporting with arrays $A=a_1,a_2,...,a_n$ and $B=b_1,b_2,...,b_n$ such that for every $i \in [n]$ $a_i,b_i \in [n^{2+\epsilon-\delta}]$ must either have $\Omega(n^{2-\delta-o(1)})$ preprocessing time or have $\tilde{\Omega}(n^{\delta-o(1)} +out)$ query time, where $out$ is the output size, unless the 3SUM Conjecture is false.
\end{theorem}

\newpage

\section{Algorithms for Solving SetIntersection}\label{sec:algorithms_for_setintersection}
 The Strong SetIntersection Conjecture argues that any solution to SetIntersection such that the query time of the solution is $O(T+out)$, where $out$ is the output size, must use $S=\tilde{\Omega}(\frac{N^2}{T})$ space. In the following subsections we present three simple algorithms that demonstrate how to achieve the $S \cdot T = O(N^2)$ tradeoff. This tradeoff is superior to the tradeoff of Cohen and Porat~\cite{CP10} and Cohen~\cite{Cohen10} where the output size is large (see the discussion in Section~\ref{sec:3sum_clb_sd_si}).

\subsection{Algorithm 1}
We are given an instance of SetIntersection with sets $S_1,S_2, ...,S_m$. Let $S_{k_1},S_{k_2},...,S_{k_p}$ be all the sets in the input instance that their size is larger than $r$, for some integer $r \geq 0$. In the preprocessing phase we create a $p \times p$ matrix $M$, such that the $i$th row and column represent the set $S_{k_i}$. At location $(s,t)$ of matrix $M$, for $1 \leq s,t \leq p$, we save a list of all elements in the intersection of $S_{k_s}$ and $S_{k_t}$. Moreover, for every set $S_i$, for $1 \leq i \leq m$, we save a hash table that contains all the elements in $S_i$. Given a query pair $(i,j)$, if one of the sets $S_i$ or $S_j$ has no more than $r$ elements then we can go over each of the elements of this set and check in the hash table of the other set if this element is also contained in the other set. If both sets have at least $r$ elements then the intersection can be found at location $(i,j)$ of matrix $M$.

\textbf{Analysis}. If at least one of the sets has no more than $r$ elements then the query time is $O(r)$ using the hash table of the other set. If both sets have at least $r$ elements then the answer is kept in matrix $M$, so the query time is $O(out)$. That is, the total query time is $O(r+out)$. The total number of elements in all sets is $O(N)$. Therefore, the size of the hash tables is $O(N)$. The number of sets with at least $r$ elements is no more than $\frac{N}{r}$. The size of the intersection between two sets that one of them has no more than $2r$ elements is bounded by $2r$. Therefore, the total space in $M$ for the rows and columns representing sets that their size is within the range $[r,2r]$ is at most $(\frac{N}{r})^2\cdot 2r = \frac{2N^2}{r}$. Using the same argument, the the total space in $M$ for the rows and columns representing sets that their size is within the range $[2^ir,2^{i+1}r]$, for $0 \leq i \leq \log{\frac{N}{r}}-1$ is at most $(\frac{N}{2^ir})^2 \cdot 2^{i+1}r = \frac{2N^2}{2^ir}$. Consequently, the total space used by the algorithm is $\sum_{i=0}^{\log{\frac{N}{r}}-1}{\frac{2N^2}{2^ir}}= \frac{2N^2}{r}\sum_{i=0}^{\log{\frac{N}{r}}-1}{\frac{1}{2^i}} = O(\frac{N^2}{r})$.

\subsection{Algorithm 2}
The idea is similar to the solution of Cohen and Porat~\cite{CP10} with some modifications.

We create a binary tree. The tree has $\log{r}+1$ levels for some $r \geq 1$.

Let $C=\{S_{k_1},S_{k_2},...,S_{k_p}\}$ be the collection of all the sets in the input instance that their size is larger than $r$. In the root we create a $p \times p$ boolean matrix such that the $i$th row and column represent $S_{k_i}$. For every $s,t \in [p]$, we set $M[s,t]$ to 1, if $S_{k_s} \cap S_{k_t} = \emptyset$, and to 0, otherwise. For the next level downward, we ignore all sets with less than $r$ elements. That is, we continue just with the sets in $C$. Let $e_1,e_2,..,e_q$ be all the elements in the sets of $C$. For every $1 \leq i \leq q$, let $f_i$ be the number of sets in $C$ that contain $e_i$. Let $z$ be the largest integer such that $\sum_{i=1}^{z}{e_i} \leq \frac{N}{2}$. The left child of the root handles all the sets in $C$ ignoring all their elements in $\{e_{z+1},...,e_q\}$, while the right child of the root handles all the sets in $C$ ignoring all their elements in $\{e_{1},...,e_{z+1}\}$. The element $e_{z+1}$ is kept in the root. It is clear that the number of elements that each node handles is no more than $\frac{N}{2}$.

We continue the construction recursively downward the tree. That is, a node $v$ in the $i$th level (the root level is 0) represents a collection $C'$ of sets $S'_1,S'_2,...,S'_{m'}$. For all the sets that their size is at least $\frac{r}{2^i}$, a matrix is created that contains the answers to disjointness queries of two sets with at least $\frac{r}{2^i}$ elements. Then, we create two child nodes. Within these child nodes, we continue with a collection $C''$ that contains all sets in $C'$ with at least $\frac{r}{2^i}$ elements. Let $e'_1,e'_2,..,e'_{q'}$ be all the elements in the sets of $C''$. For every $1 \leq i \leq q'$, let $f'_i$ be the number of sets in $C''$ that contain $e'_i$. Let $z'$ be the largest integer such that $\sum_{i=1}^{z'}{e'_i} \leq \frac{N}{2^i}$. The left child of $v$ handles all the sets in $C''$ ignoring all their elements in $\{e'_{z'+1},...,e'_{q'}\}$, while the right child of $v$ handles all the sets in $C''$ ignoring all their elements in $\{e'_{1},...,e'_{z'+1}\}$. The element $e_{z'+1}$ is kept in $v$. It is clear that the number of elements that each child node handles is no more than $\frac{N}{2^i}$.

In any node $u$ of the $(\log{r})$-level of the binary tree, the matrix we create does not contain just answers to disjointness queries, instead, it contains the complete answers to the intersection queries for each pair of the sets that $u$ represents. The nodes in the $(\log{r})$-level have no child nodes.

Finally, besides the binary tree, for every set $S_i$, for $1 \leq i \leq m$, we save a hash table $T_i$ that contains all the elements in $S_i$.

The query is handled as follows. Given a query pair $(i,j)$, we start handling the query from the root of the tree we have constructed. If either $S_i$ or $S_j$ has less than $r$ elements, we can answer the query immediately using the hash table $T_i$ or the hash table $T_j$. If both sets have at least $r$ elements we look at the proper location in the matrix in the root node in order to know if $S_i \cap S_j =\emptyset$ or not. If the intersection is empty we are done. Otherwise, we first check if the element $e_{z+1}$ is contained in both sets, if so it is reported. Then, we continue the search recursively in the two child nodes. Within the search in the left child node we ignore all the elements of the sets $S_i$ and $S_j$ that are in $\{e_{z+1},...,e_q\}$, while in the search in the right child we ignore all their elements in $\{e_{1},...,e_{z+1}\}$. Finally, if we reach a node in the $(\log{r})$-level then we report the elements of the intersection of the relevant subsets of $S_i$ and $S_j$ that remain when reaching that node. This is done by using the matrix that is saved in that node.

\textbf{Analysis}. In the $i$th level of the binary tree the total number of elements that are propagated to each node in this level is no more than $\frac{N}{2^i}$. Therefore, the number of sets in a node in the $i$th level that their size is at least $\frac{r}{2^i}$ is at most $\frac{\frac{N}{2^i}}{\frac{r}{2^i}}=\frac{N}{r}$. Consequently, the matrix size in each node of the binary tree is bounded by $(\frac{N}{r})^2$. In the $(\log{r})$-level of the binary tree we keep the complete intersection between each of the sets in a node. The total number of elements in a node in the $(\log{r})$-level is at most $\frac{N}{2^{\log{r}}}=\frac{N}{r}$. The number of sets in each node in the $(\log{r})$-level is at most $\frac{N}{r}$. Therefore, the size of the matrix in each node in the $(\log{r})$-level is also $(\frac{N}{r})^2$, as the worst case scenario is when the same element occurs in all sets. In this case, it appears in the intersection of all pairs of sets, that is, it appears at most $(\frac{N}{r})^2$ times. All in all, the size of the matrix in all nodes is no more than $(\frac{N}{r})^2$. In the $i$th level there are $2^i$ nodes. Thus, the total size of the binary tree is $\sum_{i=1}^{\log{r}}{2^i(\frac{N}{r})^2}=\tilde{O}(\frac{N^2}{r})$.

The query time in each node during the traversal of the binary tree is $O(1)$ if both sets in question have at least $\frac{r}{2^i}$ elements. Otherwise, it is $O(\frac{r}{2^i})$ as for each element of the small set we check if it is in the other set using the hash table. If one of the sets has less than $\frac{r}{2^i}$ in a node in the $i$th level the search in that path is stopped. This node is called a stopper node. Let us denote by $x$ the number of stopper nodes. Let $v_1,v_2,...,v_x$ be the stopper nodes and let us denote by $\ell_i$ the level of the stopper node $v_i$. The total query time is dominated by the query time in the stopper nodes. Searching until the stopper node is done in $O(\log{r})$ time and it is guaranteed that at least one output element is found in the stopper node. Therefore, the query time is $\tilde{O}(\sum_{i=1}^{x}{\frac{r}{2^{\ell_i}}})=\tilde{O}(r\sum_{i=1}^{x}{\frac{1}{2^{\ell_i}}}=\tilde{O}(r)$. The last equality follows from Kraft inequality that states that for any binary tree $T$ $\displaystyle\sum_{\ell \in leaves(T)}{}{2^{depth(\ell)}} \leq 1$. If we get to a node $v$ in the $(\log{r})$-level then we report all the elements of the intersection that are saved at the proper location in the matrix of $v$. Therefore, the time we spend in such node is proportional to the size of the output which is optimal. All in all, the total query time is $\tilde{O}(r+out)$.

\subsection{Algorithm 3}
In the preprocessing phase, we first calculate the frequency of each element that occurs in at least one set of the input collection. We save a list $L$ of the $r$ most frequent elements in the input sets. Moreover, we create a dictionary $D$. We create an entry $(i,j)$ in the dictionary for each pair of sets $S_i$ and $S_j$ that their intersection is not empty. This entry contains all the elements in $S_i \cap S_j$ excluding those that are in $L$. Additionally, for every set $S_i$ we save a hash table to quickly verify the existence of some element in the set. Given a query pair of sets $S_i$ and $S_j$, we first go over all the elements in $L$ and check for each one of them if it is in both $S_i$ and $S_j$ using their hash tables. Then, we output all the element in $D$ in the entry that corresponds to $S_i$ and $S_j$ if there is such an entry.

\textbf{Analysis}. The query time is obviously $O(r+out)$, as the number of elements in $L$ is $r$ and all the elements that are found in the dictionary are part of the output. Let $e_1, e_2, ..., e_k$ be all the elements in the input sets sorted in non-increasing order of frequency. Additionally, let us denote by $f_i$ the number of sets in which an element $e_i$ occurs. If an element $e_i$ does not appear in the list $L$ then its frequency $f_i$ is at most $N/r$. The number of intersections of two sets in which element $e_i$ occurs is no more than $f_i^2$. Therefore, the total size of the dictionary $D$ is $O(f_{r+1}^2+f_{r+2}^2+...+f_k^2)$. We know that $f_{r+1}+f_{r+2}+...+f_k \leq N$, and that each element is the sum is bounded by $N/r$. The sum $f_{r+1}^2+f_{r+2}^2+...+f_k^2$ is maximized when all elements in this sum are as larger as possible. Consequently, $f_{r+1}^2+f_{r+2}^2+...+f_k^2 \leq r(N/r)^2 = N^2/r$. Therefore, the total space complexity of the algorithm is $O(N^2/r)$.

\subsection{Hybrid Solution}
In the previous subsections we presented several algorithms that solves SetIntersection with $O(T+out)$ query time (where $out$ is the output size) and $S=O(\frac{N^2}{T})$ space. Cohen~\cite{Cohen10} demonstrated a solution that uses $O(N^{2-2t})$ space and answer queries in $O(N^tout^{1-t}+out)$ time for $0 \leq t \leq 1/2$. As mentioned in Section~\ref{sec:3sum_clb_sd_si}, the last solution is better than the first one whenever $out<N^{\frac{t}{1-t}}$. Otherwise, the first solution is better. It is desired to obtain a solution that combine these two tradeoffs. The simple idea to do so is as follows. We fix a specific amount of space $S$ that we are allowed to use. Then, we maintain 3 data structures that each one of them uses $S$ space: (1) Data structure for SetDisjointness that for every pair of set $S_i$ and $S_j$ from the input collection does not just answer whether their intersection is empty or not, but rather returns the size of their intersection. This can be easily done using the same space-time tradeoff that is known for SetDisjointness. That is, $S \times T = O(N^2)$, where $S$ is the space complexity and $T$ is the query time. The idea is to maintain a matrix for all sets that their size is larger than $T$ that contains an explicit answer for the size of the intersection of two sets that are both larger than $T$. For the intersection of sets that at least one of them is smaller than $T$, we can just go over all the elements of the small set and check how many of them occurs in the larger set using a hash table for the large set. (2) A data structure for solving SetIntersection using $S$ space with $O(\frac{N^2}{S}+out)$ query time. This data structure can be constructed using any of the 3 algorithms that are described in the previous subsections. (3) A data structure for solving SetIntersection using $S$ space with $O(Nout^{1-t}/\sqrt{S}+out)$ query time. This data structure can be obtained using the solution by Cohen~\cite{Cohen10} (see also the discussion in Section~\ref{sec:3sum_clb_sd_si}).

Using these 3 data structures, when we are given as a query two sets $S_i$ and $S_j$ we can find the size of their intersection by the first data structure. Then, if $out<N^{\frac{t}{1-t}}$ (the value of $t$ is fully determined by the fact that $S=O(N^{2-2t})$) we use the data structure from (3). Otherwise, we use the data structure from (2).

\end{document}